\title{Resampling-based multi-resolution false discovery exceedance control} 

\author{Jesse Hemerik\footnote{Econometric Institute, Erasmus University, Burg. Oudlaan 50,
3062 PA Rotterdam, The Netherlands. e-mail: hemerik@ese.eur.nl}
}

\documentclass[11pt]{article}

\usepackage{amsmath}
\usepackage{amsfonts}
\usepackage{amssymb,mathrsfs}
\usepackage{bm}
\usepackage{bbm}
\usepackage{amsthm}
\usepackage{comment}
\usepackage[nottoc]{tocbibind}
\usepackage{graphicx}

\usepackage{subcaption}

\usepackage{changepage}
\usepackage{natbib}
\usepackage{abstract}
\usepackage[colorlinks=true,citecolor=blue,runcolor=blue,linkcolor=blue, pdfborder={0 0 0}]{hyperref}

\usepackage{tikz}
\usepackage{pgfplots}

\usepackage{algorithm}
\usepackage[noend]{algorithmic}

\theoremstyle{plain}

\newtheorem{theorem}{Theorem}[section]

\newtheorem{proposition}[theorem]{Proposition}
\newtheorem{corollary}[theorem]{Corollary}

\theoremstyle{remark}

\newtheorem{assumption}{Assumption}[section]

\newtheorem{remark}{Remark}[section]
\newtheorem{example}{Example}[section]

\usepackage[margin=3cm]{geometry}

\newcommand{\reals}{\mathbb{R}}
 \newcommand{\eqd}{\,{\buildrel d \over =}\,}

\newcommand{\N}{\mathcal{N}}
\newcommand{\F}{\mathcal{F}}
\newcommand{\R}{\mathcal{R}}
\newcommand{\Se}{\mathcal{S}}
\newcommand{\K}{\mathcal{K}}
\newcommand{\M}{\mathcal{M}}
\newcommand{\I}{\mathcal{I}}

\newcommand{\G}{\mathcal{G}}
\newcommand{\X}{\mathcal{X}}
\newcommand{\A}{\mathcal{A}}
\newcommand{\E}{\mathcal{E}}
\newcommand{\Hy}{\mathcal{H}}
\newcommand{\eps}{\epsilon}
\newcommand{\pr}{\mathbb{P}}
\newcommand{\gr}{\mathbb{G}}

\newcommand{\citt}{\citet}
\newcommand{\citp}{\citep}

\newcommand{\fin}{\text{lim}}

\begin{document}
\maketitle

\begin{abstract}
MaxT is a highly popular resampling-based multiple testing procedure, which controls the Familywise Error Rate (FWER) and is powerful under dependence. This paper generalizes maxT to what we term ``multi-resolution'' False Discovery eXceedance (FDX) control. Basic FDX control means ensuring that the FDP ---    the proportion of false discoveries among all rejections --- is at most $\gamma$ with probability at least $1-\alpha$. Here $\gamma$ and $\alpha$ are prespecified, small values between 0 and 1. The proposed method is in addition simultaneous, in the following way: the procedure outputs a single rejection threshold $q$, but ensures that with probability $1-\alpha$,  simultaneously over all stricter thresholds, the corresponding FDPs are also below $\gamma$. In particular, for a small set of hypotheses, the FDP bound is 0, i.e., the FWER is 0. Despite these additional, simultaneous guarantees, our method has power comparable to Romano-Wolf, the most powerful non-simultaneous FDX method. Further, our method is valid under the same assumptions. Thus, this paper shows that FDX methods can often be made simultaneous almost for free. The proposed method can be formulated as an extension of simultaneous approaches such as Hemerik, Solari and Goeman (2019), for the first time allowing for  confidence envelopes with a data-dependent shape --- thus resolving a major limitation of such methods.
\\
\\
\emph{keywords:}  Bootstrap; False discovery exceedance; False discovery proportion;  Permutation; Post hoc inference; Simultaneous inference
 \end{abstract}

\section{Introduction}
Prominent ways of adjusting for multiple testing are  \emph{Familywise Error Rate} (FWER) control and controlling the mean or a quantile of  the \emph{False Discovery Proportion} (FDP) \citp{goeman2014multiple,cui2021handbook}. The FWER is the probability that there are one or more false positive findings. Controlling the FWER means ensuring that this probability is at most some prespecified value $\alpha\in(0,1)$. The FDP is the proportion of false positives among all rejected or selected hypotheses. 

Controlling a quantile of the FDP is usually called \emph{False Discovery eXceedance} (FDX) control, or sometimes simply \emph{FDP control} \citp{van2004augmentation,korn2004controlling,korn2007investigation,romano2005exact, lehmann2005generalizations, romano2006stepup,romano2006stepdown,
romano2007control,guo2007generalized,farcomeni2009generalized,roquain2011type,guo2014further,sun2015false,
delattre2015new,dohler2020controlling}. FDX control means ensuring that the FDP is at most $\gamma$ with probability at least $1-\alpha$, with $\alpha$ and $\gamma\in[0,1)$ prespecified.
There are also  methods that provide simultaneous confidence bounds for FDPs \citp{genovese2006exceedance,goeman2011multiple,rosenblatt2018all,goeman2019simultaneous,hemerik2019permutation,
blanchard2020post,katsevich2020simultaneous,vesely2023permutation,durand2025fast}. These methods can in principle also be used for FDX control, although they are not optimized for it. Indeed, since they provide a range of simultaneous FDP bounds, they do not focus their power on a particular FDP bound $\gamma$, like FDX methods.

 Since in FDX methods both $\alpha$ and $\gamma$ can be chosen, each FDX approach represents a broad range of possible methods. For example, note that when $\gamma=0$, one obtains FWER control.  Further, when $\alpha=0.5$, one ensures that $\pr(FDP\leq \gamma)\geq 0.5$, which means that the \emph{median} of the FDP is kept below $\gamma$ \citp{hemerik2025nonparametric}. This is the same as controlling the mean of the FDP (the \emph{false discovery rate}), except that the median instead of the mean is controlled. This illustrates that FDX methods are very versatile.

\subsection{Open challenges}
Developing FDX methods is challenging, especially when the tested variables are correlated, which can lead to large variability of the FDP.
Indeed, as \citt{delattre2015new} note, when no assumptions on the dependence structure are made whatsoever, existing FDX methods ``are in general too conservative''.
 \citt{lehmann2005generalizations} provide p-value-based FDX control, which is valid under mild assumptions on the dependence structure of the p-values.   However, for many dependence structures, this method is rather conservative. The reason is that this approaches takes into account worst-case scenarios for the dependence structure, instead of adapting to the dependence structure. 
 
 It is known that resampling-based multiple testing methods can adapt to the dependence structure in the data and thereby achieve better power than competitors \citp{meinshausen2006false,hemerik2019permutation,blain2022notip,andreella2023permutation}.
 Resampling-based methods can be based on permutation or bootstrap-type tests. Using permutations can lead to exact properties for finite samples, while bootstrap methods are  always asymptotic.
The most popular resampling-based multiple testing method is the maxT procedure \citp{westfall1993resampling,westfall2008multiple}, which controls the FWER.
 The idea behind this procedure is that the $(1-\alpha)$-quantile of the maximum of the test statistics corresponding to the true hypotheses, depends on the dependence structure of the test statistics. Using permutations or bootstrapping, the procedure constructs a rejection threshold that is adapted to this dependence structure.
If our goal is FWER control --- i.e., FDX control with $\gamma=0$ --- then  maxT  is often a good choice \citp{meinshausen2011asymptotic}. 
 
The maxT method is very popular in e.g. neuroimaging \citp{nichols2002nonparametric,eklund2016cluster,white2023controlling}, 
economics \citp{romano2005stepwise,romano2010hypothesis,romano2016efficient,list2019multiple,list2023multiple,
batista2025mobile,deserranno2025promotions,dupas2025informing} and 
omics research \citp{steiss2012permory,terada2015high,rosyara2016software,power2016genome,saffari2018estimation}.
However, often researchers are willing to incur a few false positives if this leads to many more true discoveries. This means that it often worthwhile to take  $\gamma>0$, i.e. to consider general FDX control. A method that allows for this is the procedure in \citt{korn2004controlling} and \citt[][Algorithm 4.1]{romano2007control}. This method is often referred to as the (step-down) Romano-Wolf procedure and is further studied in \citt{delattre2015new}. Romano-Wolf  does not necessarily use resampling, but is usually most powerful when combined with resampling.
The method then uses permutations or bootstrapping to construct a threshold and rejects all hypotheses with test statistics above this threshold. The threshold is adapted to the dependence structure and consequently the procedure is powerful. 
Romano-Wolf is asymptotically valid, and finite-sample valid if the test statistics corresponding to the true hypotheses are independent from the statistics of the false hypotheses --- although usually the method works well without that assumption \citp{delattre2015new}. 
Romano-Wolf is defined very differently from the method that we will propose, but they turn out to have similar power for FDX control (see Section \ref{secsims}) and the new method is valid under the same assumptions. However, the new method additionally provides simultaneous, \emph{multi-resolution} control, as discussed below.

The  method of \citt{hemerik2019permutation} and \citt{andreella2023permutation}  can also be used for resampling-based FDX control, although it is not optimized for it, as explained above. This procedure provides FDP bounds that are simultaneous over multiple rejection thresholds. By choosing the most lenient threshold for which the FDP bound is still below $\gamma$, we can obtain FDX control.
 Table \ref{tab:overview} provides an overview of previously existing resampling-based procedures  that provide confidence on the FDP.

 \begin{table}[h]
\centering
\caption{Overview of existing resampling-based multiple testing methods that provide confidence on the FDP. Like the new method, these procedures are compatible with the permutation and bootstrap tests referenced in this paper.  
}
\label{tab:overview}
\begin{tabular}{p{7cm} p{7.5cm}}  
\hline
Statistical criterion & Recommended method \\
\hline
Familywise error rate (FWER) control & MaxT \citp{westfall1993resampling,romano2005stepwise}  \\[6pt]

Provide a confidence bound for the FDP, given a rejection threshold  & \citt{hemerik2018false}, which extends \citt{tusher2001significance}  \\[17pt]

False discovery exceedance (FDX) control & \citp[][Algorithm 4.1]{romano2007control}; see \citt{delattre2015new} for more properties \\[6pt]

Simultaneous confidence bounds for all FDPs & The methods in \citt{hemerik2019permutation,blain2022notip,andreella2023permutation} are strongly related. \citt{vesely2023permutation} may be more powerful when there are many signals   \\
\hline
\end{tabular}
\end{table}

\subsection{Contributions of this paper}
This paper presents a resampling-based method that  requires exactly the same input as the maxT method: a matrix of resampled test statistics. The output of the new method is a single real number $q$. Rejecting all hypotheses with test statistics exceeding $q$, guarantees FDX control. In addition, the method guarantees that with probability $1-\alpha$, for all $t\geq q$ the FDP for threshold $t$ is at most $\gamma$. Thus, the method provides a particular type of simultaneous control, which has not been considered before in the literature. This is what we call \emph{multi-resolution FDX control}.
 It means that the method allows ``zooming in'' after seeing the data, or ``drilling down'', as it is called in \citt{rosenblatt2018all} and  \citt{andreella2023permutation}. For example, suppose that 42 test statistics exceed the threshold $q$, i.e., we can reject 42 hypotheses, while controlling the FDX rate. Then, after seeing the data, we can additionally decide to zoom in on e.g. the 20 hypotheses with the largest test statistics. If $\gamma=0.05$, then we  know that among these 20 hypotheses, at most $5\%$, so  one  hypothesis, is true. In fact, we also know this about the top 39 hypotheses, since $5\%$ of $39$ is less than 2. Moreover, we know that among the top 19 hypotheses, none are true.  All these statements will be simultaneously valid with probability at least $1-\alpha$. 
(See Remark \ref{remarksimul} for a related example.) 

Our term ``multi-resolution'' is on purpose similar to the term ``all-resolutions'' used in \citt{rosenblatt2018all}. That paper provides bounds for all FDPs of all sets of hypotheses. This method is usually less powerful than \citt{hemerik2019permutation}, which is in turn less powerful for multi-resolution FDX control than the method proposed here.  
The reason is that  \citt{hemerik2019permutation} is not able to focus its power on a particular FDP bound $\gamma$. This is because that method
uses a confidence envelope with a shape that is not allowed to depend on the data. 
If their \emph{candidate confidence envelopes} depend on the data, their proof  breaks down fundamentally.
The new approach however, can be formulated as a generalization of \citt{hemerik2019permutation} to certain confidence envelopes with data-dependent shapes, as explained in the Supplementary Material.

Further, the new method is more user-friendy than \citt{hemerik2019permutation}, since it involves no tuning parameters, such as $\mathbb{B}$, $\mathbb{T}$ and $s$. Moreover, its results are easier to report. 
 Further, it turns out that the new procedure is strongly connected to maxT in other ways. First of all, like maxT, the new method has a single-step and a sequential version. For $\gamma=0$, these coincide with the single-step and sequential versions of maxT. As we increase $\gamma$ however, the power of the new method strongly increases, since we then allow for a small fraction of false positives. There are also other connections with maxT. In particular, if $\gamma>0$  and the new method rejects fewer than $\gamma^{-1}$ hypotheses, then with probability 1, maxT rejects exactly the same  hypotheses.  However, when maxT rejects at least $\gamma^{-1}$ hypotheses, then the new method usually rejects many more. These connections with maxT are proved in the Supplementary Material. 

Like maxT, the new procedure can be combined with many permutation  and bootstrap tests, many of which are  asymptotically exact \citp{pollard2003resampling,winkler2014permutation,hemerik2020robust,desantis2025inference}.
This means that the new method can be used for testing hypotheses in e.g. generalized linear models with many responses and several nuisance covariates. 
This paper primarily focuses on finite-sample properties, but our method can be combined with such asymptotic tests exactly like maxT \citp[as in e.g.][]{de2025permutation}.
 Our procedure is implemented in the R package \verb|rFDP| \citp{hemerik2025rFDP}. This package takes a general matrix of resampled statistics as input.

The structure of this paper is as follows. In Section \ref{secsetting} we define notation, formulate assumptions and discuss the set $\G$ of transformations, which is a subset of a group $\gr$. Section \ref{secsingle} defines the single-step variant of our  method. The sequential variant is defined in Section \ref{secfullseq}, with Section \ref{secap} defining a faster approximation of that method.  
Sections \ref{secsims} through \ref{secdisc} provide simulation results, analyses of real data and the Discussion. 
The Supplementary Material contains more theoretical results, additional data analyses and proofs.


\section{Setting and notation} \label{secsetting}

Consider  data $X$  taking values in sample space $\X$ and hypotheses $\Hy_1,...,\Hy_m$. Just like maxT, our method can be combined with various permutation and bootstrap tests. Hence, the hypotheses could take various forms and may concern basic case-control, linear \citp{winkler2014permutation}, generalized linear \citp{desantis2025inference,de2025permutation} or time series models \citp{romano2005stepwise}.
Consider corresponding test statistics $T_1(X),...,T_m(X)$, taking values in $\mathbb{R}$ or a subset of $\mathbb{R}$. For example, each $T_i(X)$ might be the absolute value of a t-statistic. 
Let 
$$\N:=\{1\leq i \leq m: \Hy_i \text{ is true}\},$$
$$\F:=\{1\leq i \leq m: \Hy_i \text{ is false}\},$$  where we assume for convenience that $\N$ is not empty and $m\geq 2$.

\subsection{Invariance assumption} \label{secinvas}

For finite-sample validity, we need an invariance assumption. We will also  discuss asymptotic validity, which is possible under weaker assumptions. 

Consider a set $\gr$ of transformations  $g: \X\mapsto \X$, for example a set of maps that permute, rotate or sign-flip data.  We assume $\gr$ is a \emph{group} with respect to composition of maps \citp{hoeffding1952large,hemerik2018exact,dobriban2025symmpi}. Instead of using the full group of e.g. permutations, it can be faster and sometimes even more powerful to use a cleverly chosen subgroup, which is of course also a group \citp{chung1958randomization,koning2024more,koning2024morep}.
 Write $gX$ instead of $g(X)$ for short. 
Consider the following assumptions, the second of which implies the first.

\begin{assumption} \label{assjointd}
The joint distribution of the test statistics $T_i(gX)$ with $i \in \N$ and $g\in \gr$ is invariant under all transformations $g\in \gr$ of $X$. 
\end{assumption}

\begin{assumption} \label{assjointdextra}
Conditional on $(T_i(X):i\in \F)$, Assumption \ref{assjointd} holds.
\end{assumption}

Throughout the paper  we assume Assumption \ref{assjointdextra} holds. This is always the case if Assumption \ref{assjointd} holds and 
the data corresponding to $\F$ and $\N$ are independent. This latter assumption is also required for the Romano-Wolf method to have proven finite-sample validity \citp{delattre2015new}. It is worth remarking that in our simulation studies, we found no example where the weaker Assumption \ref{assjointd} was not sufficient. This suggests that either Assumption \ref{assjointdextra} is unnecessary or it may be substantially  weakened. Further, we show in the Supplementary Material that asymptotically we can drop Assumption \ref{assjointdextra}.
To our knowledge, there exists no powerful, finite-sample FDX method that makes no assumption regarding dependencies between the test statistics or p-values. 


Assumption \ref{assjointd} is standard in the literature on finite-sample valid permutation-based multiple testing  \citp{meinshausen2006false,goeman2010sequential, hemerik2018false, hemerik2019permutation, blanchard2020post}.  
In  a simple case-control study, for this assumption to hold it is sufficient that the joint distribution of the variables corresponding to $\N$ is the same for the cases as for the controls. Assumption \ref{assjointdextra} would be satisfied if this still holds conditional on the observations corresponding to $\F$.
If we only require asymptotic exactness, then Assumption \ref{assjointd} is not necessary for some resampling-based multiple testing methods \citp{pollard2003resampling}. In particular, there is much literature on bootstrap tests in combination with maxT \citp{pollard2003resampling,romano2005stepwise,romano2010hypothesis,romano2016efficient,list2019multiple,list2023multiple}. Instead of Assumption \ref{assjointd}, one then essentially needs that asymptotically, one can sample from the distribution of the test statistics $T_i(X)$ with $i \in \N$. The new method is then also asymptotically valid. Since our method can be combined with various permutation and bootstrap tests and the method in \citt{hemerik2020robust} and \citt{desantis2025inference}, it can be used for testing in linear and generalized linear models for example.

For finite samples, permutation methods tend to be reasonably robust to violations of Assumption \ref{assjointd}.
Robustness of permutation-based multiple testing methods against violations of Assumption \ref{assjointd} is discussed in \citt{anish}. In unbalanced case-control studies, bootstrap tests may provide more reliable performance than permutation tests \citp{pollard2003resampling}
Further, certain  semiparametric approaches, based on e.g. sign-flipping,  that can deal well with heteroscedasticity. These approaches can be combined  with  maxT and likewise with the new method \citp{davidson2008wild, winkler2014permutation,hemerik2020robust,desantis2025inference}. In particular, \citt{de2025permutation} illustrates how to combine sign-flipping of score contributions with maxT; combining this sign-flipping approach with the new method is analogous, since the new method requires the same input as maxT.
Finally, finite-sample and asymptotic robustness of single-hypothesis permutation tests is discussed in e.g. \citt{romano1990behavior, canay2017randomization, kashlak2022asymptotic}.

\subsection{Equivalence classes of transformations}

Sometimes $\gr$ contains subsets of equivalent transformations, which always give the same test statistic. For example, in a case-control study, if we shuffle the controls among themselves, this usually does not change the test statistic \citp[][p.815]{hemerik2018exact}.
This relates to the following assumption, which is typically satisfied whenever the data are continuous. We make this paper throughout the paper and supplement.

\begin{assumption} \label{asscont}
There exists a partition $\gr_1,...,\gr_d$ of $\gr$ with $|\gr_1|=...=|\gr_d|$ such that with probability 1, for all $g, g'\in \gr$ and $1\leq i \leq m$, $T_i(gX)=T_i(g'X)$ if and only if $g$ and $g'$ are in the same set $\gr_j$. Further,  with probability 1, for all $1\leq i, i' \leq m$, $T_i(gX)\neq T_{i'}(g'X)$ if $g$ and $g'$ are in different sets $\gr_j$.
\end{assumption}

The above assumption guarantees in particular that instead of using the full group $\gr$, we can use one transformation from each subset $\gr_j$:  for all methods considered, this will be equivalent to using the full group $\gr$ but faster. 
Let $\G\subseteq \gr$ be such that  it contains exactly one element from each of $\gr_1,...,\gr_d$. Also, assume without loss of generality that $\G$ contains the identity map $id$. 
Whether  $|\G|<|\gr|$, depends on the setting:
for two-group comparisons, typically we will have $|\G|<|\gr|$, \citp[][p.815]{hemerik2018exact}. In case of independence testing of two continuous variables, typically all permutations give different statistics, so that we have $|\gr_1|=...=|\gr_d|=1$, i.e., $\G=\gr$. If using all transformations in $\G$ is computationally infeasible, we could use random transformations as in e.g. \citt{hemerik2018false}.

\subsection{Rejection sets and FDP}

Given data $x\in\X$, the set of rejected hypotheses that we will consider is of the form 
$$\R(t,x)=\{1\leq i \leq m: T_i(x)>t\},$$
 where $t\in\mathbb{R}$.  
 Here $t$ is the \emph{rejection threshold}. 
 Note that the rejection region is simply $(t,\infty)$ for every test statistic.
 This  could straightforwardly be extended to general rejection regions $D_i(t)\subset\mathbb{R}$, where the region $D_i(t)$ shrinks as $t$ increases and $D_i(\cdot)$ potentially depends on the hypothesis index $i$.

For $t\in\mathbb{R}$ and $x\in \X$ we let
\begin{align*}
R(t,x) &= |\R(t,x)|, \\
V(t,x) &=|\N\cap\R(t,x)|,  \\
FDP(t,x) &=    V(t,x)/(R(t,x)\vee 1).  
\end{align*}
Note that $R(t,x)$ is the number of rejections given data $x$. When we write e.g. $R(t)$ or $FDP(t)$, this will be  understood to mean  $R(t,X)$. $V$ is the usual notation for the number of false discoveries in the literature.
Note that the function $t\mapsto FDP(t, x)$ is piecewise constant and right-continuous.
Finally, let $\alpha\in(0,1)$ and $\gamma \in[0,1)$. 

\section{Single-step procedure} \label{secsingle}
In this section, we define  the \emph{single-step procedure}. This method is uniformly improved by the sequential procedure in Section \ref{secseq}. 
The sequential procedure directly builds on the single-step procedure, which forms its first step.
The Supplementary Material discusses an algorithm for the single-step method.
While we initially define the method based on test statistics, it is also possible to use p-values as input. This is  also discussed in the  Supplementary Material.

The threshold $q$ of the single-step method is defined as follows. Recall the definition of $\G$, which is a set of representative transformations from the group $\gr$.
For every $g\in \G$, let 
\begin{equation} \label{eqdefsg}
s_g=s_g(X) :=\sup\Big\{t\in\mathbb{R}: \frac{R(t,gX)}{R(t,X)\vee 1}>\gamma\Big\}.
\end{equation}
This supremum is always a real number, which follows from the fact that for all small enough $t$, the fraction in \eqref{eqdefsg} is $m/m=1$ and for all large enough $t$, the fraction is $0/(0\vee 1)=0$.
Let $q=q(X)$ be the $(1-\alpha)$-quantile of the values $s_g$, $g\in \G$ --- more precisely, 
$$q := \min\big \{t\in\mathbb{R}: \frac{|\{g\in \G: s_g\leq t\}|}{|\G|}\geq 1-\alpha \big\}.$$

The following theorem --- see the Supplementary Material for proofs --- states that if we use $q$ as the rejection threshold, then we have FDX control. 
The intuition behind this is the following. For  $t\in \reals$ and $g\in \G$, $R(t,gX)$ can be seen as a conservative estimate of the number of false positives $V(t,X)$ \citp[this is the intuition behind the SAM methodology, see][]{tusher2001significance,hemerik2018false}.
Hence,
${R(t,gX)}/{R(t,X)\vee 1}$ can be seen as a conservative estimate of $FDP(t,X)$. This means that $s_g$ can be seen as an estimate of the largest threshold $t$ for which  $FDP(t,X) >\gamma$. Thus, beyond $s_g$, the FDP is estimated to be at most $\gamma$. Finally, to have $(1-\alpha)100\%$ confidence, the method takes the $(1-\alpha)$ quantile of the values $s_g$. The actual reasoning is much more complex; the required proof is much longer than e.g. the proofs in  \citt{hemerik2019permutation}. What makes the argument difficult is that it does not rely on \emph{candidate envelopes} that are independent of the data. This is further explained in the Supplementary Material.

The theorem does not only guarantee FDX control, but ensures multi-resolution FDX control:
with probability $1-\alpha$, simultaneously over all $t\geq q$, we have $FDP(t)\leq \gamma$. 
Since the FDP bounds that the procedure provides are all $\gamma$ (roughly speaking), the method allows for relatively easy reporting of results: given $T_1,...,T_m$, if we simply know $q$, then all FDP statements follow from it. When desired, the FDP statements about the sets $\R(t)$, $t\geq q$, can be expanded to \emph{any} subsets of hypotheses using \emph{interpolation}, as already covered in e.g. \citt{blanchard2020post}. Interpolation is also behind e.g. \citt{van2004augmentation} and \citt{farcomeni2009generalized} and  all-resolutions inference methods for neuroimaging \citp{blain2022notip, andreella2023permutation}.


\begin{theorem} \label{thmmain}
The single-step method provides FDX control:
\begin{equation} \label{eq:weaker}
\pr\big\{FDP(q)\leq \gamma\big\}\geq 1-\alpha.
\end{equation}

In fact, we have the stronger result that
\begin{equation} \label{eq:stronger}
\pr\big\{\forall t\geq q: FDP(t)\leq \gamma\big\}\geq 1-\alpha.
\end{equation}
\end{theorem}

\begin{remark} \label{remarksimul}
Theorem \ref{thmmain} implies that
$$\pr\big[\bigcap_{t\geq q} \{V(t)\leq \gamma R(t) \}   \big]\geq 1-\alpha.$$
In particular, if we pick $t'\geq q$ such that $R(t')<\gamma^{-1}$, then $\pr\{V(t')=0\}\geq 1-\alpha$. Thus, simultaneously with the FDX control, we obtain FWER control. Because of the simultaneity,  it will hold for instance that
$$\pr\Big[\big\{FDP(q)\leq \gamma \big\}\cap \big\{V(t')=0\big\} \Big]\geq 1-\alpha.$$
As a  specific example, suppose $\gamma=0.1$, $R(q)> 9$ and  there are no ties among the test statistics $T_1(X),...,T_m(X)$. We can then pick $t'\geq q$ such that  $R(t')=9<\gamma^{-1}$. 
Then we know that with probability at least $1-\alpha$, we do not only have  $FDP(q)\leq \gamma$ but also $V(t')=0$.
\end{remark}


\begin{remark}
The new method achieves good power by adapting to the unknown dependence structure in the data, in the following way. Roughly speaking, we repeatedly simulate draws from the joint null distribution of the test statistics. In the new method, the quantities $s_g$, $g\in \G$, are all based on such ``simulated draws''. Taking the $(1-\alpha)$-quantile leads to a confidence bound that is adapted to the dependence structure. In this way the procedure avoids accounting for worst-case scenarios for the dependence structure.
\end{remark}

The threshold $q$ is a simple quantile. The underlying quantities $s_g$, $g\in \G$, however are suprema. It is discussed in the Supplementary Material how these suprema may be conveniently computed in an exact way. Further, the Supplementary Material
provides an equivalent formulation of our method  in terms of p-values instead of test statistics.

Finally, the Supplementary Material shows that asymptotically, Assumption \ref{assjointdextra} is not necessary. As remarked,  Assumption \ref{assjointd} is not necessary either for asymptotic validity, but it should hold in an asymptotic sense: asymptotically we should be able to resample from the joint distribution of the test statistics $T_i(X)$, $i \in \N$. 

\section{Sequential method} \label{secseq}

For several multiple testing approaches, e.g, Bonferroni  \citp{holm1979simple}, there exists a single-step version and a sequential version. Other examples are the  maxT method, its generalization to $k$-FWER control \citp[][Algorithm 2.1]{romano2007control} and the method in \citt{hemerik2019permutation}. Sequential versions of multiple testing methods are usually uniformly more powerful than the corresponding single-step versions.
The Romano-Wolf FDX method \citp[][Algorithm 4.1]{romano2007control} is always sequential. That method  uses a $k$-FWER method in every step. When a sequential $k$-FWER procedure is used within every step, then Romano-Wolf is thus doubly sequential. It is worth noting that this does not apply to the proposed  method, whose construction is very different from Romano-Wolf.

In Section  \ref{secfullseq} we provide a sequential method that uniformly improves the single-step method of Theorem \ref{thmmain}. The sequential procedure produces a threshold $q_{\fin}\leq q$, which is potentially strictly smaller than $q$, which means that it can lead to more rejections. We show that if we replace $q$ by $q_{\fin}$ in Theorem \ref{thmmain}, then the theorem still holds. 
Compared to the  famous  sequential methods Bonferroni-Holm and maxT,  the sequential version of our method is relatively complex and computationally intensive.
In Section \ref{secap} we provide a faster method that approximates the full sequential method.

\subsection{Full sequential method} \label{secfullseq}
Before we formally state the result in Theorem \ref{thmseq}, we explain the idea behind the construction of $q_{\fin}$.
Firstly, recall that our single-step method  rejects all hypotheses with $T_i(X)>q$. 
However, the proof of Theorem \ref{thmmain} also defines a potentially strictly smaller threshold $q'\leq q$. That proof shows that if we substitute $q'$ for $q$ in Theorem \ref{thmmain}, then the result still holds true. However, computing $q'$ requires knowledge of $\N$, so we cannot know it. However, we can potentially close in on $q'$.

Indeed, in  the proof of Theorem \ref{thmmain} an event $\E$ was defined which satisfies $\pr(\E)\geq 1-\alpha$. 
We showed that under $\E$, $FDP(q)\leq \gamma$.  This means that under $\E$, there are at least $B_1:=\lceil (1-\gamma)R(q)\rceil $ false hypotheses among the hypotheses with indices in $\R(q)$. Thus, under $\E$, we can choose a set 
$\I\subseteq\{1,...,m\}$ with $\I^c\subseteq  \R(q)$ and $|\I^c| = B_1$, such that $\N\subseteq\I$. Here  $\I^c$ stands for $\{1,...,m\}\setminus\I$. We cannot pinpoint such a set $\I$ but we know it exists, under $\E$.

For such a set $\I$, we know that for every $g\in \G$ and $t\in \reals$,
$$|\I\cap \R(t,gX)| \geq |\N\cap \R(t,gX)| = V(t,gX)$$
and hence 
$$s_{g,1}^{\I}:= \sup\Big\{t\in\reals:  \frac{|\I\cap \R(t,gX)|}{R(t,X)\vee 1}>\gamma\Big\} \geq s_{g}' :=\sup\Big\{t\in\reals:  \frac{V(t,gX)}{R(t,X)\vee 1}>\gamma\Big\}.$$
Recall from the proof of  Theorem \ref{thmmain}  that $q'$ is defined as the $(1-\alpha)$-quantile of the values $s_g'$, $g\in \G$.
Hence, we know that 
$q^{\I}_1\geq q'$, where $q'$ is the $(1-\alpha)$-quantile of the values $s_{g}'$, $g\in \G$.     
We do not know the set $\I$, but we can instead use the maximum over all such sets $\I$, i.e., we use
$$q_1   = \max\big\{ q^{\I}_1 :  \I\subseteq\{1,...,m\}, \I^c\in \R(q) \text{ and } |\I^c| = B_1   \big\}.$$
This threshold can be computed and satisfies $q\geq q_1$. 
It can be validly used, i.e., we may replace $q$ by $q_1$ in Theorem \ref{thmmain}.  
Now we can repeat the above process, but starting from $q_1$ instead of $q$. Continuing like this, we find a sequence $q\geq q_1\geq q_2\geq ...$ of thresholds that can be computed and validly  used. This sequence converges after a finite number of steps and we take $q_{\fin}$ to be the limit. The following theorem states that this final threshold $q_{\fin}$ can indeed be used. A formal proof is provided in the Supplementary Material.


\begin{theorem} \label{thmseq} 
 Let $q_0=q$.   For  $i \in \{1,2,...\}$, sequentially define $q_{i}$ based on $q_{i-1}$ as follows. Let
$B_{i}=\lceil (1-\gamma)R(q_{i-1})\rceil $ and write 
$$\K_i = \big\{  \I\subseteq\{1,...,m\}: \I^c\subseteq \R(q_{i-1}) \text{ and } |\I^c| = B_i    \big\}.$$
For $\I\in \K_i$ let
$$s_{g,i}^{\I}= \sup\Big\{t\in\reals:  \frac{|\I\cap \R(t,gX)|}{R(t,X)\vee 1}>\gamma\Big\}$$
and denote by $q^{\I}_i$  the $(1-\alpha)$-quantile of the values $s_{g,i}^{\I}$, $g\in \G$. 
Define 
\begin{equation} \label{defqi}
q_i   = \max\big\{ q^{\I}_i :  \I\in\K_i   \big\}.
\end{equation}
Then $q\geq q_1\geq q_2\geq...$
 and from some $i\in \mathbb{N}$ we have $q_i=q_{i+1}=...$ . Define $q_{\fin}$ to be this final value. Then Theorem \ref{thmmain} still holds if we replace $q$ by $q_{\fin}$. 
\end{theorem}

In the sequential method, in every step $i$ we need to compute $s_{g,i}^\I$ for every $\I \in \K_i$ and $g\in \G$. It is discussed in the Supplementary Material how this can be done in a convenient way.

 \subsection{Approximation of the sequential method} \label{secap}
In step $i=1$ of the sequential method, we must compute compute $q_1$, which can be computationally intensive, since for $i\geq 1$, $q_i$ is a maximum over all sets $\I\in \K_i$, where 
 $$|\K_i| = \binom{R(q_{0})}{B_i} = \binom{R(q_{0})}{ \lceil (1-\gamma )R(q_{0})\rceil }.$$
 Assuming one takes $0<\gamma\leq 0.5$,  it may be computationally infeasible to perform this  step when $\gamma R(q_{0})$ is large. As a practical solution, rather than computing a maximum over all $\I\in \K_1$ in step 1, we can compute the maximum over some randomly chosen subset of $\K_1$, say $\Se_1 \subseteq{\K_1}$.
 This leads to an approximation $\hat{q}_1 $, which is smaller than or equal to $q_1$.
 
 In steps $2, 3,...$ we can proceed analogously. More precisely, take $\hat{q}_{0}=q$, let $i\geq 1$  and suppose we have just 
 computed $\hat{q}_{i-1}$. 
  Then, we let $\hat{B}_i  = \lceil (1-\gamma)R(\hat{q}_{i-1}) \rceil$ and let $\Se_i$ be a random subset of  
$$ \hat{\K}_i := \big\{  \I\subseteq\{1,...,m\}: \I^c\in \R(\hat{q}_{i-1}) \text{ and } |\I^c| = \hat{B}_i    \big\}.$$
More precisely, the set $\Se_i$ collects $M$  uniform draws from  $\hat{\K}_i$, where $M$ is some large user-specified number. The draws could be defined to be with or without replacement; in this paper we draw with replacement.
Then, for every $\I\in \Se_i$, we compute
$$\hat{s}_{g,i}^{\I}= \sup\Big\{t\in\reals:  \frac{|\I\cap \R(t,gX)|}{R(t,X)\vee 1}>\gamma\Big\}$$
and denote by $\hat{q}^{\I}_i$  the $(1-\alpha)$-quantile of the values $\hat{s}_{g,i}^{\I}$, $g\in \G$. Then, we take 
$$\hat{q}_i:= \max\{\hat{q}^{\I}_i: \I\in \Se\}.$$
We compute $\hat{q}_1, \hat{q}_2,...$ in this manner until we reach a step where the threshold no longer decreases. Alternatively, we could stop after a predetermined number of steps. In either case, we then  define  $\hat{q}_{\fin}$ to be the last obtained threshold, which serves as our approximation of the threshold $q_{\fin}$ from Section \ref{secseq}.

This defines the \emph{approximation method}, which is also given Algorithm \ref{a:ap}.  This algorithm uses a computational strategy analogous to that of the single-step method, see section B of the  Supplementary Material. 
 The algorithm uses all transformations $g\in \G$, but to further increase speed these may be replaced by a random sample from $\G$, as already mentioned in Section \ref{secsetting}. The method is implemented in the R package \verb|rFDP| \citp{hemerik2025rFDP}.

\begin{algorithm}[ht!] 
\caption{Approximation of the sequential method}
\begin{algorithmic}  \label{a:ap}
\STATE $q \gets$ result of single-step method (see Algorithm 1 in Supplementary Material)
\STATE $\hat{q}_0 \gets q$  
\STATE $i \gets 1$  
\WHILE{$i=1$ OR $\hat{q}_{i-1}<\hat{q}_{i-2}$}  
	\STATE $\hat{B}_i \gets \lceil (1-\gamma)R(\hat{q}_{i-1}) \rceil$
	\STATE $\Se_i \gets$ a large random subset of $\big\{  \I\subseteq\{1,...,m\}: \I^c\in \R(\hat{q}_{i-1}) \text{ and } |\I^c| = \hat{B}_i    \big\}$
	\FORALL{$\I\in \Se_i $}
		\FORALL{$g\in \G$}
		\STATE $\M_g^{\I} \gets \big\{T_i(gX): i\in \I \big\}\cup \big\{T_i(X): 1\leq i \leq m\big\}.$
		\STATE $\hat{s}_{g,i}^{\I-} \gets  \max\{t\in \M_g^{\I}: \frac{| \I\cap \R(t,gX)    |}{R(t,X)\vee 1}  >\gamma \}$
		\STATE $\hat{s}_{g,i}^{\I} \gets  \min\{t\in \M_g^{\I}: t > \hat{s}_{g,i}^{\I-}\}$
		\ENDFOR
		\STATE $\hat{q}_i^{\I} \gets$ $(1-\alpha)$-quantile of the values $\hat{s}_{g,i}^{\I}$, $g\in \G$
	\ENDFOR
	\STATE $\hat{q}_i   \gets \max\big\{ \hat{q}^{\I}_i :  \I\in\Se_i    \big\}$
	\STATE $i \gets i+1$
\ENDWHILE
\STATE $\hat{q}_{\fin} \gets \hat{q}_{i-1}$
\RETURN $\hat{q}_{\fin}$
\end{algorithmic}
\end{algorithm}

\section{Simulations} \label{secsims}
We performed simulations to compare the new method with the  most relevant competitors, which are known to provide powerful FDX control: the most powerful method among the two FDX procedures in \citt{lehmann2005generalizations} (which is uniformly more powerful than \citealp{romano2006stepdown}); the resampling-based  method  \citt[][Algorithm 4.1]{romano2007control}; the resampling-based procedure from  \citt{hemerik2019permutation} and the method from \citt{katsevich2020simultaneous}.
We pick these methods because they are known to be more powerful than various other competitors.
It is worth remarking that the method in \citt{andreella2023permutation} is directly based on \citt{hemerik2019permutation}; a consequence is that these procedures have exactly the same power for FDX control. The method in \citt{rosenblatt2018all} typically has lower power than \citt{andreella2023permutation}. The method in \citt{blain2022notip} is the same as \citt{andreella2023permutation}, except that it uses additional external data to find a potentially more powerful set of candidate envelopes.



Details on the methods used are provided below. All methods provided valid FDX control in our simulation settings, and the new method provided valid multi-resolution FDX control. Hence, below, we focus on power comparisons.

\subsection{Simulation settings and competitors} \label{secsimsetting}
We simulated  two-group data with a total sample size of 20 and $m$ variables observed per individual.
The data were a $20$-by-$m$ matrix, with the first 10 rows representing the first group and the last 10 rows the second group.
The rows were independent of each other, and in each row there was a homogeneous correlation $\rho$ between the variables. For other correlation structures, the relative performances of the methods were comparable to those shown in the figures below.
 The observations were standard normally distributed, except that  in the columns corresponding to the false hypotheses, we added a signal $d>0$ to the observations of the first group. Note that normality is not at all required for our method to be exact and is just used as an example. 
 The null hypotheses were $\Hy_1,...,\Hy_m$, where $\Hy_i$ is the hypothesis that the observations in column $i$ are identically distributed. The corresponding test statistics $T_1,...,T_m$ were absolute values of two-sample t-statistics. We took $\alpha=\gamma=0.1$.
Besides $\rho$, we varied $\pi_0$, which is the number of true hypotheses divided by $m$.

 The first competitor is the  most powerful method among the two FDX procedures in \citt{lehmann2005generalizations}, which we will simply refer to as Lehmann-Romano.
This method takes a vector of p-values as input and is valid under mild assumptions on their dependence structure. The p-values were based on two-sided t-tests. The Lehmann-Romano method is implemented in the R package  \verb|FDX| \citp{dohler2024FDX}.

 The second competitor is the resampling-based Romano-Wolf method \citp[][Algorithm 4.1]{romano2007control}. The k-FWER method used within the Romano-Wolf procedure was the usual k-FWER generalization of the  maxT method \citp{romano2006stepdown}. Like our method, this procedure is based on permutation or bootstrap tests. In this case, permutations of the group labels were used.
 For computational feasibility, we sampled 50 random permutations in each simulation. Using more permutations would increase computation times and would barely improve power in our experience \citp[also see e.g.][]{marriott1979barnard}. 
 

The third competitor is the approach in \citt{hemerik2019permutation}, which also uses resampling to adapt to the dependence structure. Within this method we also used 50 random permutations.
The procedure provides simultaneous confidence bounds for the FDP for a range of p-value thresholds $t\in\mathbb{T}\subseteq[0,1]$. The method is theoretically compared to the new procedure in the Supplementary Material. We used the method of \citt{hemerik2019permutation}  for FDX control by choosing the largest $t\in\mathbb{T}$ for which the bound for the FDP is at most $\gamma.$ The method depends on multiple prespecified tuning parameters, namely $\mathbb{T}$ and the set $\mathbb{B}$ of so-called candidate envelopes, which is in fact an infinite-dimensional parameter. Here, we chose the tuning parameters in such a way that this method performed well for FDX control in our simulation settings. We did this to give this competitor a fair change against the new method. Specifically, we took $\mathbb{T}=[0,0.05]$ and took $\mathbb{B}$ to be the set of envelopes defined in the first paragraph of Section 2.6 of \citt{hemerik2019permutation}, with $\delta=0.001$. We observed that making $\delta$ much larger or smaller, decreased the power of that method. Moreover, other choices for $\mathbb{T}$ did not improve its power.

The last competitor is the method from \citt[][Thm 1]{katsevich2020simultaneous}, which also provides simultaneous FDP bounds. This method is elegant and fast, but is based on the assumption that every null p-value is independent of all other p-values.

\subsection{Fast  methods} \label{secsimssingle}
The new method, Romano-Wolf and \citt{hemerik2019permutation} each come in the form of a fast version and a version that is computationally intensive, especially when $m$ is large.
We first considered the fast versions, in simulations with $m=500$. 
In Section \ref{secsimsseq} we consider the computationally intensive versions.
For the new method and \citt{hemerik2019permutation}, the fast versions are  simply the single-step  versions. By the fast version of Romano-Wolf we mean that the \emph{single-step} resampling-based k-FWER method \citp[][Section 2.1]{romano2007control}   is used within Romano-Wolf. This version of Romano-Wolf is still sequential, in the sense that it sequentially applies single-step k-FWER methods \citp[see][Algorithm 4.1]{romano2007control}.

  In Fig. \ref{fig:plots_alpha0p1_legend} the simulation results can be seen for $m=500$ and varying $\rho$ and $\gamma$.
The figure shows the power for FDX control of the method from Section \ref{secsingle} and the competitors --- where \emph{power} is defined as the average fraction of the false hypotheses that was rejected.  
 We see that the Romano-Wolf method had slightly better ``power'' than the new method --- if we ignore that the new method provides many additional confidence statements.  The present power analysis does not take into account that the new method provides simultaneous conclusions, which makes Romano-Wolf seem slightly superior, if we only look at the figure. 
 
  We further see that Lehmann-Romano had the lowest power, especially under dependence. After that, the method from  \citt{hemerik2019permutation} often had the lowest power.
 It should also be added that the method from \citt{hemerik2019permutation} usually did not provide any FWER rejections, meaning there was no set of rejections for which the FDP bound was 0. For our method on the other hand,  the top $\gamma^{-1}-1$ rejections were always FWER rejections. Moreover, the novel method was more convenient since we did not need to specify tuning  parameters. Choosing good tuning parameters for \citt{hemerik2019permutation} is especially problematic if one does not know a priori whether there are few or many signals. 
 
 The power of \citt{katsevich2020simultaneous} was comparable to the new method in some settings. Note  however that the method of \citt{katsevich2020simultaneous} is based on the assumption that every null p-value is  independent. Consequently,  the method does not provide generally valid simultaneous bounds. Moreover, that method had relatively low power when $m$ was smaller, as illustrated in the next section.

\begin{figure}[ht!] 
\centering
  \includegraphics[width=\linewidth]{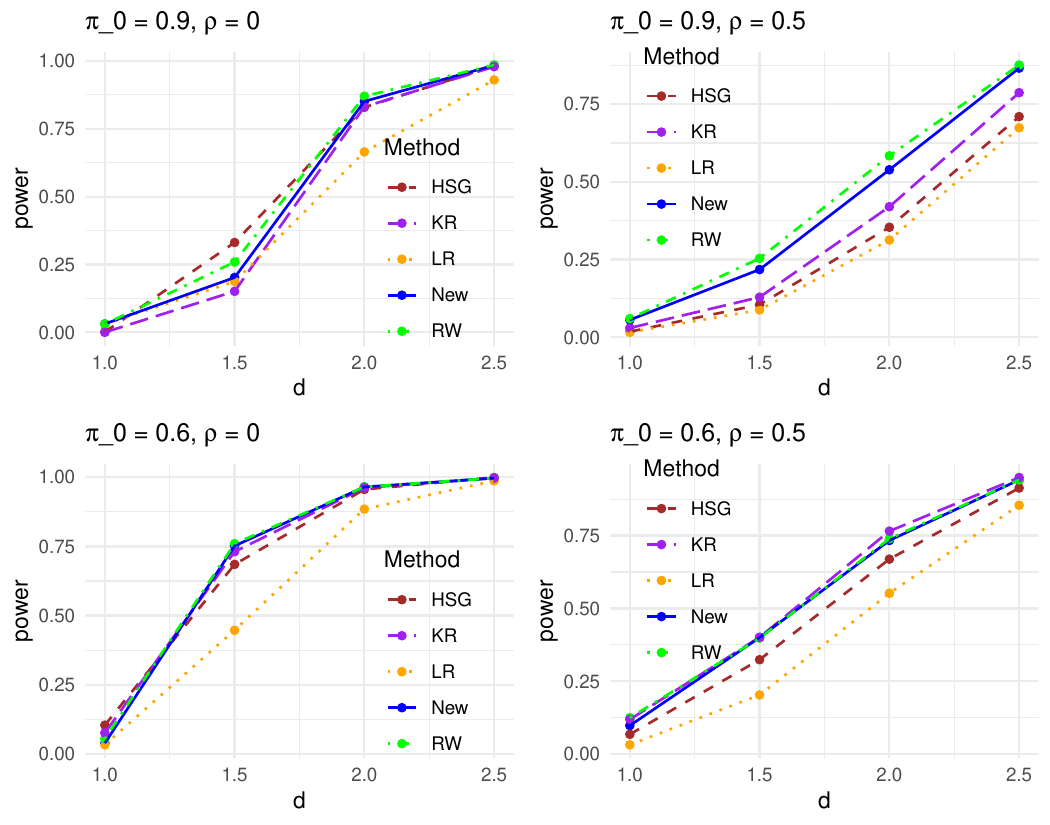}
  \caption{The power of the new method (``New'') versus Lehmann-Romano (``RS''), \citt{hemerik2019permutation} (``HSG''),  Romano-Wolf (``RW'') and Katsevich-Ramdas (``KR''), for $m=500$ and varying $\pi_0$, $\rho$ and signal size $d$. Every estimate is based on $10^3$ repeated simulations.}
 \label{fig:plots_alpha0p1_legend}
\end{figure}

\subsection{Computationally intensive methods}  \label{secsimsseq}
We now illustrate the performance of the computationally intensive versions of the methods, for $m=100$. We also include the method of \citt{katsevich2020simultaneous} for comparison.  Apart from that method, we show results for six methods: the fast and the intensive versions of the new method, \citt{hemerik2019permutation} and Romano-Wolf. The intensive, i.e., sequential versions of the new method, \citt{hemerik2019permutation} and Romano-Wolf
 are often computationally infeasible but can be approximated using random combinations (see Section \ref{secap}), which is what was done here, with the number of combinations  in each step set to 25. Using random combinations in Romano-Wolf means  that in step 2 and onwards of the sequential k-FWER method, we do not check all combinations but 25 randomly chosen ones.
Taking the number of combinations   higher than 25 did usually not yield a different result in these methods. The sequential method from \citt{hemerik2019permutation} further requires choosing an additional tuning parameter $s$, which we took to be $0.005$ as in the simulations of \citt{hemerik2019permutation}.

The results are in Fig. \ref{fig:plots_sequential_m100_legend_RW}.   We took $\pi_0=0.6$, since  for $\pi_0=0.9$, sequential multiple testing methods only provide marginal improvements over their single-step variants. The results illustrate that the sequential variant of each method improves its single-step version. As in Section \ref{secsimssingle},
 Romano-Wolf had slightly better ``power'' than the new method, but the new method provides additional, simultaneous guarantees. In other words, the new method is more powerful when we start zooming in, as already discussed in e.g. the Introduction and Section \ref{secsingle}.

\begin{figure}[ht!] 
\centering
  \includegraphics[width=\linewidth]{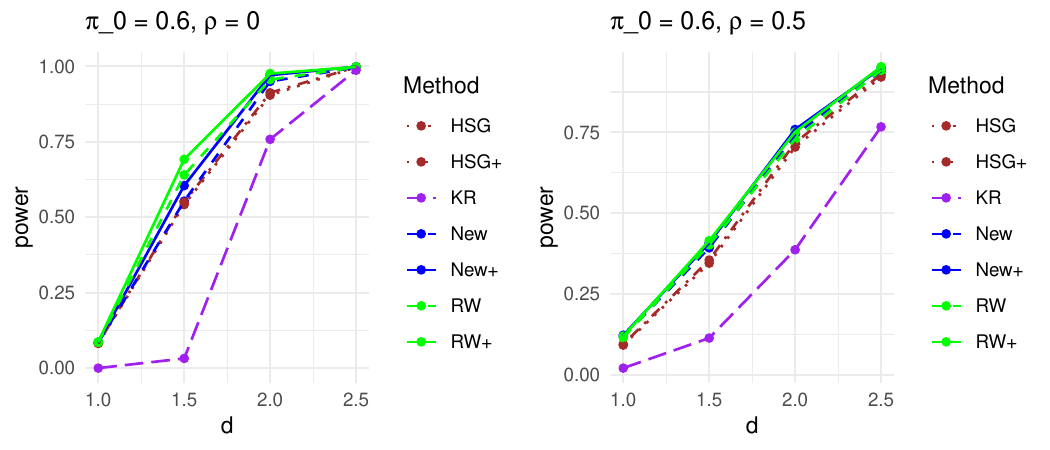}
  \caption{The power of the fast (``New'', ``HSG'', ``RW'') and computationally intensive versions (``New+'', ``HSG+'', ``RW+'') of the new method,  \citt{hemerik2019permutation} and Romano-Wolf,  for $m=100$ and varying $\rho$ and signal size $d$. Results for the Katsevich-Ramdas method (``KR'') are also shown. Every estimate is based on $10^3$ repeated simulations.}
 \label{fig:plots_sequential_m100_legend_RW}
\end{figure}

\section{Real data analysis} \label{secdataribo}
The Supplementary Material contains an extensive analysis of aggregated music review data. Here, we provide a brief illustrative  analysis of a  dataset about riboflavin (vitamin B2) production with \emph{B. subtilis}. This dataset
is freely available \citp{buhlmann2014high}. It contains normalized measurements
of expression rates of 4088 genes from $n = 71$ samples. 
Further, the dataset contains the corresponding 71 measurements of the logarithm of the riboflavin production
rate.
For each $1 \leq i \leq  4088$, we are interested in the hypothesis  $\Hy_i$ that the riboflavin production
rate was independent
of the expression level of gene $i$. 
 We took $\alpha=0.05$. Within every method, we used $10^3$ random permutations of the production rates.
 For every gene and permutation, we computed the test statistic as the  absolute value of the Pearson correlation between that gene's expression rates and the permuted riboflavin production
rates.
 
 For the maxT method to be valid, we need Assumption \ref{assjointd} to hold. For discussions on why methods based on this assumption are reasonable in practice, see    \citt{anish}, \citt[][\S5]{hemerik2019permutation}  and Section \ref{secinvas} of this paper.
First we applied the single-step maxT method, which rejected 74  hypotheses. The sequential maxT method   led to the same  rejections. 

For our FDX method to be theoretically valid for finite samples, we need Assumption \ref{assjointdextra} to hold. Although this conditional invariance assumption may not hold, in our simulations studies we found no setting where this assumption was necessary, as already mentioned; further, the assumption is not necessary asymptotically, as shown in the Supplementary Material.
We took $\gamma=0.1$ and applied our simultaneous FDX method from Section \ref{secsingle}, which rejected 186 hypotheses. This means that with $95\%$ confidence, we know that among the 186 hypotheses with the largest test statistics, at least $(1-\gamma)100\%=90\%$, so 168, are false. Simultaneously, we know that among the 19 hypotheses with the largest test statistics, at least $90\%$, so 18, are false. Likewise, we simultaneously know that among the 9 hypotheses with the largest test statistics, at least $90\%$, so all 9, are false. Thus, is we reject the top 9 hypotheses, then we have FWER control. All these statements are simultaneously valid with probability at least $1-\alpha=0.95$.

\section{Discussion} \label{secdisc}

What existing simultaneous FDP methods \citp{goeman2019simultaneous,hemerik2019permutation,
blanchard2020post,katsevich2020simultaneous,vesely2023permutation}  have in common, is that they do not excel at FDX control, since they cannot focus power on a given target FDP like FDX methods or FDR methods. 
Vice versa,  existing FDX methods are not simultaneous.
This paper shows however, that we can have powerful FDX control while also providing a useful type of simultaneous control. Indeed, we have seen that for FDX control, the new method is almost as powerful as Romano-Wolf, the most powerful basic FDX method. At the same time, the new method provides many additional guarantees, namely simultaneous FDX control at multiple resolutions.   In particular, if the user ``zooms in'' far enough, they will obtain FWER-control statements, simultaneously with the other FDP statements. 
The reason why we do not have much lower power for FDX control than Romano-Wolf, is that the FDP is typically roughly decreasing as a function of the rejection threshold, and the new method exploits this.  
Importantly, it has proven validity under the assumptions under which the resampling-based Romano-Wolf method is known to be valid --- although these two methods are defined very differently otherwise.

To our knowledge, this is in fact the first paper that compares the performance of FDX methods with simultaneous FDP methods. For example, the resampling-based simultaneous FDP papers \citt{hemerik2019permutation},   \citt{blain2022notip},  \citt{andreella2023permutation} and \citt{vesely2023permutation} do not compare their performance with the resampling-based basic FDX method of \citt{romano2007control}. Although Romano-Wolf is not simultaneous, it is still interesting to fix a $\gamma$ and see how the simultaneous methods perform compared to Romano-Wolf. 



From a user's perspective, an important advantage of the new method is that it does not require choosing tuning parameters except $\alpha$ and $\gamma$, unlike e.g. \citt{hemerik2019permutation}. 
Further, the output of the new method is simply a single number $q$, or $q_{\fin}$ for the sequential version. Despite this simple output, the procedures provide simultaneous statistical guarantees.  
 The new methods require the same input as the maxT procedure and have been implemented in the R package \verb|rFDP| \citp{hemerik2025rFDP}.
 Differences with \citt{hemerik2019permutation}, \citt{blain2022notip} and  \citt{andreella2023permutation} on a more fundamental methodological level are discussed in the Supplementary Material.


In this paper, we have mostly focused on proving finite-sample properties. Of course, finite-sample validity is only possible in simple models. For example, tests in generalized linear models with continuous nuisance covariates are never exact. However, if we only require asymptotic exactness, the new method can definitely be combined with resampling-based asymptotic tests. For example, we can combine the method with the score flipping approach from \citt{hemerik2020robust,de2025permutation}, in the same way as the maxT method can be combined with it \citp[see][]{desantis2025inference}. This allows for testing in generalized linear models for example.

\setlength{\bibsep}{3pt plus 0.3ex}  
\def\bibfont{\footnotesize}  

\bibliographystyle{biblstyle}
\bibliography{bibliography}

\appendix

\section{Overview of the Supplementary Material}

The Supplementary Material is structured as follows.

\begin{itemize}

\item  Section  \ref{seccomp} discusses how the rejection thresholds $q$ and $q_{\fin}$ can be conveniently computed in practice, in an exact way.

\item  Section \ref{secformpvs} provides an equivalent formulation of our methods in terms of p-values rather than test statistics. 

\item Section \ref{comphsg} provides a theoretical comparison of the new method with \citt{hemerik2019permutation} and related methods.

\item  Section \ref{appas} discusses asymptotic control, showing that Assumption \ref{assjointdextra} is not necessary asymptotically.

\item Section \ref{secconmaxt} discussed connections between the proposed method and maxT.

\item Section  \ref{appct} discusses the new methods in the context of \citt{goeman2021only}, which provides general theory on multiple testing procedures that make confidence statements on false discovery proportions.

\item Section \ref{secdataan} contains an  analysis of a dataset with aggregated music review scores.

\item  Section \ref{appproofs} contains proofs of results from the paper and this supplement

\end{itemize}



\section{Computing the rejection threshold} \label{seccomp}
\subsection{Computation of the threshold $q$ for the single-step method}  \label{seccompsingle}
The threshold $q$ is a simple quantile. The underlying quantities $s_g$, $g\in \G$, however are suprema.
To compute $s_g$, note that for any $x\in\X$ the function $t\mapsto R(t,x)$ is a right-continuous step function which is  non-increasing in $t$. Consequently, the function 
\begin{equation} \label{mapttofdpest}
t\mapsto \frac{R(t,gX)}{R(t,X)\vee 1} 
\end{equation}
is also  a right-continuous step function. 
For all small enough $t$ it equals 1.
Note that this function has a discontinuity at $s_g$. More precisely, the function     $t\mapsto R(t,gX)$ jumps down at $s_g$. 
 Let $\M_g$ be the set of all discontinuities of the function \eqref{mapttofdpest}, i.e.,  $$\M_g := \Big\{T_i(gX): 1\leq i \leq m\Big\}\cup \Big\{T_i(X): 1\leq i \leq m\Big\}.$$
 
Let 
\begin{equation} \label{eqsgmin}
s_g^-= \max\{t\in \M_g: \frac{R(t,gX)}{R(t,X)\vee 1}  >\gamma \}
\end{equation}
where the maximum of an empty set is $-\infty$. Note that with probability 1, $s_g^-$ is strictly smaller than $\max(\M_g)$, since for   $t=\max(\M_g)$, the fraction in \eqref{eqsgmin} equals $0/(0\vee 1)=0$.
Right from $s_g^-$, $\frac{R(\cdot,gX)}{R(\cdot ,X)\vee 1}$ stays constant until it reaches the next element of  $\M_g$. Note that $s_g$ must equal this next element.
Thus, we have 
$$ s_g = \min\{t\in \M_g: t> s_g^-\}.$$
This formula is useful for computing    $s_g$ for every $g\in \G$, and hence $q$. This method  is also shown in Algorithm \ref{a:si}.

\begin{algorithm}[ht!] 
\caption{Single-step method}
\begin{algorithmic}  \label{a:si}
\FOR{$g\in \G$}
    \STATE $\M_g \gets \Big\{T_i(gX): 1\leq i \leq m\Big\} \cup  \Big\{T_i(X): 1\leq i \leq m\Big\} $
    \STATE $ s_g^- \gets \max\Big\{t\in \M_g: \frac{R(t,gX)}{R(t,X)\vee 1}  >\gamma \Big\}$
    \STATE $ s_g \gets \min\Big\{t\in \M_g: t> s_g^-\Big\}$
\ENDFOR
\STATE $q \gets$ $(1-\alpha)$-quantile of the $s_g$, $g\in \G$.
\RETURN $q$
\end{algorithmic}
\end{algorithm}

\subsection{Computation of the threshold $q_{\fin}$ for the sequential method}
In the sequential method, in every step $i$ we need to compute $s_{g,i}^\I$ for every $\I \in \K_i$ and $g\in \G$.
To compute this supremum, we can proceed in a manner that is analogous to the computation of $s_g$ in Section \ref{seccompsingle}. Namely,
let $\M_g^{\I}$ be the set of all discontinuities of the maps $t\mapsto |\I\cap \R(t,gX)|$ and $t\mapsto R(t,X)$. Thus,
$$\M_g^{\I} := \Big\{T_i(gX): i\in \I \Big\}\cup \Big\{T_i(X): 1\leq i \leq m\Big\}.$$
 We first compute $$s_{g,i}^{\I-}:= \max\{t\in \M_g^{\I}: \frac{| \I\cap \R(t,gX)    |}{R(t,X)\vee 1}  >\gamma \}.$$
 With probability 1, we have $s_{g,i}^{\I-}<\max(\M_g^{\I})$.
Then, we have $s_{g,i}^{\I} =   \min\{t\in \M^I_g: t > s_{g,i}^{\I-}\}$. 

\section{Equivalent formulation in terms of p-values} \label{secformpvs}
The input for our single-step method are the test statistics $T_i(gX)$ with $1\leq i \leq m$ and $g\in \G$.
Sometimes is may be preferred to use p-values instead of test statistics. The method can straightforwardly be reformulated in terms of p-values $P_1(X),...,P_m(X)$, which take values in $(0,1]$. We must then adjust the assumptions by replacing $T_i$ by $P_i$ in Assumptions \ref{assjointd}-\ref{asscont}.
Further, we then consider thresholds $t\in(0,1]$ and for $x\in \X$ define
$$\R(t,x)=\{1\leq i \leq m: P_i(x)<t\}.$$
Note that since our approach is nonparametric, we do not need to assume that the p-values are valid in the sense that $\pr(P_i\leq c)\leq c$ for $c\in[0,1]$ if $\Hy_i$ is true.


Note that if we have p-values $P_1(X),...,P_m(X)$, we can translate them into test statistics by simply taking e.g. $-P_1(X),...,-P_m(X)$ or $1-P_1(X),...,1-P_m(X)$. Then, our method from Section \ref{secsingle} can be used. Equivalently, we can directly use p-values, if we slightly adjust the method as follows.

When p-values are used, we can define $s_g$ as before, except that it is now an infimum: 
$$s_g=s_g(X) :=\inf\Big\{t\in(0,1]: \frac{R(t,gX)}{R(t,X)\vee 1}>\gamma\Big\}.$$
 Further, we now define $q$ to be the $\alpha$-quantile of the values $s_g$, $g\in \G$ --- more precisely,
 \begin{equation} \label{eqdefqpv}
 q := \min\{t\in\mathbb{R}: \frac{|\{g\in \G: s_g\geq  t\}|}{|\G|}\geq 1-\alpha \}.
 \end{equation}
 We reject all hypotheses $\Hy_i$ with $P_i(X)<q$.

For each $g\in \G$, to compute $s_g$ in practice, we can proceed analogously to Section \ref{seccomp}, namely as follows. 
Note that the function  $$(0,1]\ni t\mapsto \frac{R(t,gX)}{R(t,X)\vee 1} $$
is left-continuous and stepwise constant. It has a discontinuity at $s_g$. More precisely, the function     $t\mapsto R(t,gX)$, which is non-decreasing,  jumps up at $s_g$.
 Let  $$\M_g := \Big\{P_i(gX): 1\leq i \leq m\Big\}\cup \Big\{P_i(X): 1\leq i \leq m\Big\}.$$
and
$$s_g^-= \min\{t\in \M_g: \frac{R(t,gX)}{R(t,X)\vee 1}  >\gamma \},$$
where the minimum of an empty set is $\infty$.
Note that with probability 1, $s_g^-$ is strictly larger than $\min(\M_g)$,  since for   $t=\min(\M_g)$, the fraction in \eqref{eqsgmin} equals $0/(0\vee 1)=0$.
Note that we then have 
$$ s_g = \max\{t\in \M_g: t < s_g^-\}.$$  
We then compute $q$ using formula \eqref{eqdefqpv}.
This method  is also shown in Algorithm \ref{a:sip}. 
In the main paper, we consider test statistics rather than p-values, so the definitions from Section \ref{secsingle} apply.
\begin{algorithm}[ht!] 
\caption{Single-step method based on p-values}
\begin{algorithmic}  \label{a:sip}
\FOR{$g\in \G$}
    \STATE $\M_g \gets \Big\{P_i(gX): 1\leq i \leq m\Big\} \cup  \Big\{P_i(X): 1\leq i \leq m\Big\} $
    \STATE $ s_g^- \gets \min\Big\{t\in \M_g: \frac{R(t,gX)}{R(t,X)\vee 1}  >\gamma \Big\}$
    \STATE $ s_g \gets \max\Big\{t\in \M_g: t< s_g^-\Big\}$
\ENDFOR
\STATE $q \gets$ $(1-\alpha)$-quantile of the $s_g$, $g\in \G$.
\RETURN $q$
\end{algorithmic}
\end{algorithm}

\section{Theoretical comparison with \citt{hemerik2019permutation}} \label{comphsg}
Here, we explain the theoretical innovation of the new method compared to \citt{hemerik2019permutation} and the related methods in \citt{blain2022notip}  and \citt{andreella2023permutation}.
Like the proposed method, the method in \citt{hemerik2019permutation} provides certain simultaneously valid statements on FDPs --- in the form of a \emph{confidence envelope}. If a method is based on test statistics rather than p-values, then  a confidence envelope is a function $B:\reals \rightarrow \mathbb{N}$ satisfying
$$ \mathbb{P} \Big (\bigcap_{t\in \reals} \big \{V(t) \leq B(t)\big \}  \Big ) \geq 1- \alpha.$$
In \citt{hemerik2019permutation}, the function $B(\cdot)$ is picked from a monotone family $\mathbb{B}$ of \emph{candidate envelopes}. Importantly, the envelopes in the family $\mathbb{B}$ have a shape that should be chosen independently from the data. We will now discuss that the new procedure can also be seen as a method that provides a confidence envelope; however, the shape of this confidence envelope  depends on the data, in such a way that the method is powerful for (multi-resolution) FDX control.

The new method guarantees that 
$$\pr\big\{\forall t\geq q: FDP(t)\leq \gamma\big\}\geq 1-\alpha,$$
or equivalently, 
\begin{equation} \label{eq:defce}
\pr\big\{\forall t\geq q: V(t)\leq \gamma R(t)\big\}\geq 1-\alpha.
\end{equation}
Here $q$ is the threshold of the single-step method; the above still holds if we replace $q$ with $q'$, the threshold of the sequential method.
The guarantee \eqref{eq:defce} is equivalent to  the  confidence envelope $B^q(\cdot)$, where for  $s\in \reals $
\begin{equation} \label{eq:Bq}
B^s(t) =
\begin{cases}
\lfloor \gamma R(t) \rfloor & \text{if } t\geq s, \\
R(t)  & \text{if } t<s.
\end{cases}
\end{equation}

We thus see that the new method can be formulated as a procedure  that provides a confidence envelope, just like \citt{hemerik2019permutation}. In fact,  like \citt{hemerik2019permutation}, the new method can be seen as choosing the envelope from a set of candidate envelopes, namely $\mathbb{B}=\{B^s: s\in \reals\}$. They key difference with \citt{hemerik2019permutation}, \citt{blain2022notip}  and \citt{andreella2023permutation} however, is that the new candidate envelopes  $B^s$  depend on the data, since they depend on the actual numbers of rejections $R(t)$, $t\in \reals$. This is thus the key theoretical  innovation of the new method. This innovation is important, since the authors of \citt{hemerik2019permutation} were forced to make the candidate envelopes data-independent to make their theory work. That is particularly problematic  if one does know whether there are few or many signals.  With the new approach, the candidate envelopes depend on the data. In particular, we can define the candidate envelopes in such a way that we have high power for multi-resolution FDX control for a given $\gamma\in[0,1)$. Note that if there are few signals, then the new method behaves similarly to maxT (as shown in Section \ref{secconmaxt}), and when there are many signals, it tends to reject much more than maxT.  Also note that our candidate envelopes  $B^s$ do not require choosing any tuning parameters, unlike the envelopes in \citt{hemerik2019permutation} and  \citt{andreella2023permutation}. A further minor advantage is that we do not need to compute p-values, unlike those methods. A disadvantage of the new method is that  we make Assumption \ref{assjointdextra}, although we found no  settings where it was actually necessary, and asymptotically we do not need it (see Section \ref{appas}).

Note that in the definition of the new method, the suprema $s_g$ are defined as the largest value for which a certain property is satisfied. In the proposed method, this criterion is $\{R(t,gX)\}/\{R(t,X)\vee 1\}>\gamma,$ but it may  also be possible to replace this by other criteria, so that other new methods can be defined. We leave that possibility for future research.

\section{Asymptotic control} \label{appas}
Here we show that asymptotically, Assumption \ref{assjointdextra} is not necessary.
Note that when the sample size $n$ increases to infinity, usually the cardinality of $\G=\G_n$ also increases to infity. Hence, instead of using all transformations in $\G_n$, we  fix an integer $w\geq \alpha^{-1}$ independently of $n$ and use a collection $\G'_n=(g_1,...,g_w)$ of random transformations from $\G_n$, Here we define $g_1$ to be the identity, and sample $g_2,...,g_w$ with or without replacement. Using such random transformations is still valid, in the sense of \citt{hemerik2018exact,hemerik2018false}. 

\begin{proposition} \label{propasympt}
Suppose that the data $X=X_n$ depend on $n$,  which usually indicates the sample size. 
Suppose that as $n\rightarrow\infty$, for every $i\in \F$, $T_i(X_n)$ converges to $\infty$ in probability, while the  test statistics $T_i(g'X_n)$ with $i\in \N$ and $g'\in \G_n$ remain bounded in probability.
Suppose that Assumption \ref{assjointd} holds for every $n$, but do not make Assumption \ref{assjointdextra}. 

Then asymptotically we have simultaneous FDX control, in the sense that
$$\liminf_{n\rightarrow\infty} \mathbb{P}\big\{\forall t\geq q: FDP(t, X_n)\leq \gamma\big\}\geq 1-\alpha,$$
for $q$ as defined in Section \ref{secsingle}. The same is true for the sequential method from Section \ref{secfullseq}, i.e., if we replace $q$ by $q_{\fin}$.

\end{proposition}

\section{Connections with the maxT method} \label{secconmaxt}

If we take $\gamma=0$, then FDX methods guarantee that $\pr(FDP> 0)\leq \alpha$, i.e., the FWER is controlled.
The most popular resampling-based FWER controlling method is the  \emph{maxT method} (defined in Section \ref{specialcasegamma0}) by Westfall and Young \citp{westfall1993resampling,westfall2008multiple,meinshausen2011asymptotic,koning2024morep}. 

It turns out that the  methods proposed in the previous sections, are closely connected 
to maxT.
In particular, as shown in Section \ref{specialcasegamma0}, when $\gamma=0$, our single-step and sequential FDX method coincide with respectively the single-step and sequential version of maxT. 
This is quite surprising, since single-step maxT does not involve suprema like our single step method, and the sequential versions of both methods are very different as well.
Nevertheless, the new procedures are thus a generalization of maxT, allowing the user to pick any $\gamma\in[0,1)$ rather than only $\gamma=0$.  The power of the new method strongly increases with $\gamma$, since increasing $\gamma$ means allowing for some false positives. This is illustrated in Fig. \ref{fig:powervsgamma}.
The Romano-Wolf method also coincides with maxT for $\gamma=0$, which follows immediately from the definition Romano-Wolf. Showing that the new method reduces to maxT for $\gamma=0$, is more work.


\begin{figure}[ht!] 
\centering
  \includegraphics[width=0.8\linewidth]{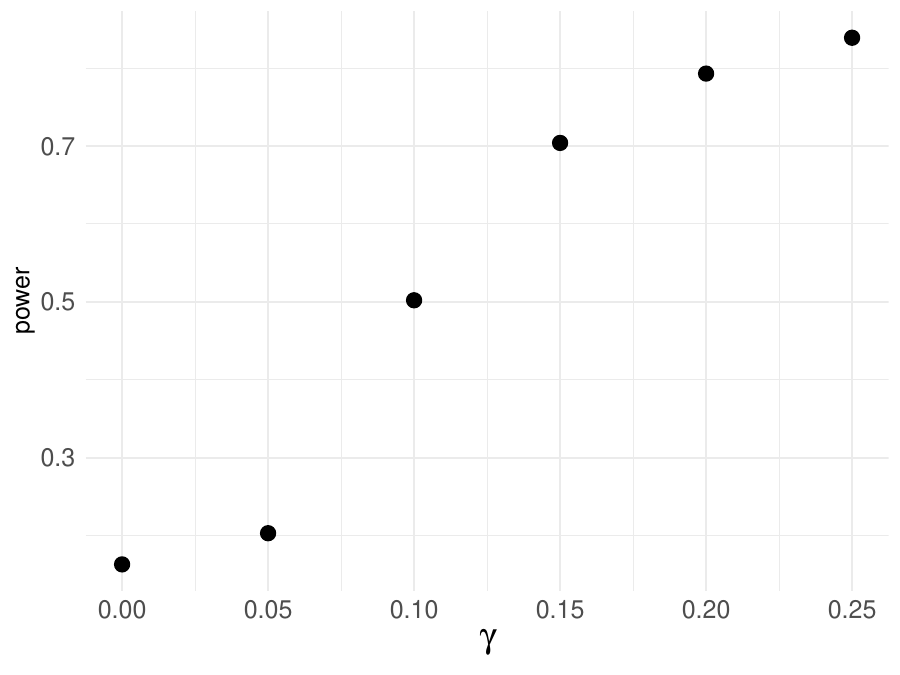}
  \caption{The power of the single-step FDX method as depending on $\gamma$. The power is the mean fraction of the false hypotheses that are rejected. For $\gamma=0$, the new method coincides with maxT. The power estimates are based on $10^3$ repeated simulations. The data were distributed as specified in Section \ref{secsimsetting}, with $m=500$, $\rho=0$ and 100 false hypotheses. The plot shows that single-step maxT --- i.e., the new method with $\gamma=0$ --- rejected about 15 false hypotheses on average. The power increases with $\gamma$.}
 \label{fig:powervsgamma}
\end{figure}

There are also other connections with maxT. In Section \ref{secotherconmaxt} we show that if $\gamma\in(0,1)$ and the maxT method rejects fewer than $\gamma^{-1}$ hypotheses, then with probability 1, our method rejects exactly the same  hypotheses. Note that our method then provides the exact same confidence statement as maxT, namely that the FDP is 0.

Further, we show that for  every $\gamma\in[0,1)$, with probability 1, our method rejects at least one hypothesis if and only the maxT methods reject at least one hypothesis. This means that the global test implied by our procedure is the same as that of maxT. In Section \ref{appct}  we discuss this in light of the closed testing theory in \citt{goeman2021only}.

\subsection{The maxT method is the special case $\gamma=0$} \label{specialcasegamma0}

The well-known maxT procedure is defined as follows.
The single-step maxT method simply rejects all hypotheses $\Hy_i$ for which $T_i(X)>Q_0$, where $Q_0$ denotes the $(1-\alpha)$-quantile of the maxima $\max_{1\leq i \leq m} T_i(gX)$, $g\in \G$. The sequential method then continues iteratively as follows. Let $\R^0$ be set of indices of hypotheses rejected by the single-step maxT method, which we call step 0. 
Next, we define 
$\R^1=\{1\leq i \leq m: T_i(X)>Q_1\}$, where  $Q_1$ denotes the $(1-\alpha)$-quantile of the maxima  $\max_{i\in (\R^0)^c} T_i(gX)$, $g\in \G$. 
We continue like this, in the $j$-th step rejecting all hypotheses with indices in $$\R^j =\{ i\in\{1,...,m\}:T_i(X)>Q_j\},$$ where  $Q_j$ denotes the $(1-\alpha)$-quantile of the maxima  $\max_{i\in  {(\R^{j-1})}^c} T_i(gX)$, $g\in \G$. 
We obtain a sequence $Q_0\geq Q_1\geq...$, which converges to $Q_{\fin}$, say, after a finite number of steps.
The sequential maxT method rejects all hypotheses $\Hy_i$ for which $T_i(X)>Q_{\fin}$.

\begin{theorem} \label{thmgamma0}
Suppose $\gamma=0$. Then 
 our 
single-step  method (Section \ref{secsingle}) rejects exactly the same hypotheses as the single-step maxT method. The same holds for our sequential  method (Section \ref{secfullseq}) and the sequential maxT method.
\end{theorem}

\subsection{Other connections with maxT} \label{secotherconmaxt}

The following proposition  states that if the new method rejects few hypotheses, then maxT rejects the same, and vice versa. This means that when there are few signals in the data, the new method roughly behaves like maxT. 

\begin{proposition} \label{samewhenfewrej}
Consider any $\gamma\in(0,1)$. Suppose that our single-step method (Section \ref{secsingle}) or single-step maxT rejects strictly fewer than $\gamma^{-1}$ hypotheses. Then with probability 1, they reject exactly the same hypotheses; more precisely, $q=Q_0$. 
The same holds for the sequential versions of our method (Section \ref{secfullseq}) and the maxT method; we then have $q_{\fin}=Q_{\fin}$ with probability 1.
\end{proposition}

Note that roughly speaking, Theorem \ref{thmgamma0} follows from Proposition \ref{samewhenfewrej}; indeed, if we take $\gamma$ close enough to 0, then Proposition \ref{samewhenfewrej} says that with probability 1, our method and maxT reject the same hypotheses. However, the statement in Theorem \ref{thmgamma0} holds surely and not only with probability 1.
That the new method rejects the same as maxT when one of the methods rejects fewer than $\gamma^{-1}$ hypotheses, is illustrated in Fig. \ref{fig:plotnrej}.


 \begin{figure}[ht!] 
\centering
  \includegraphics[width=0.8\linewidth]{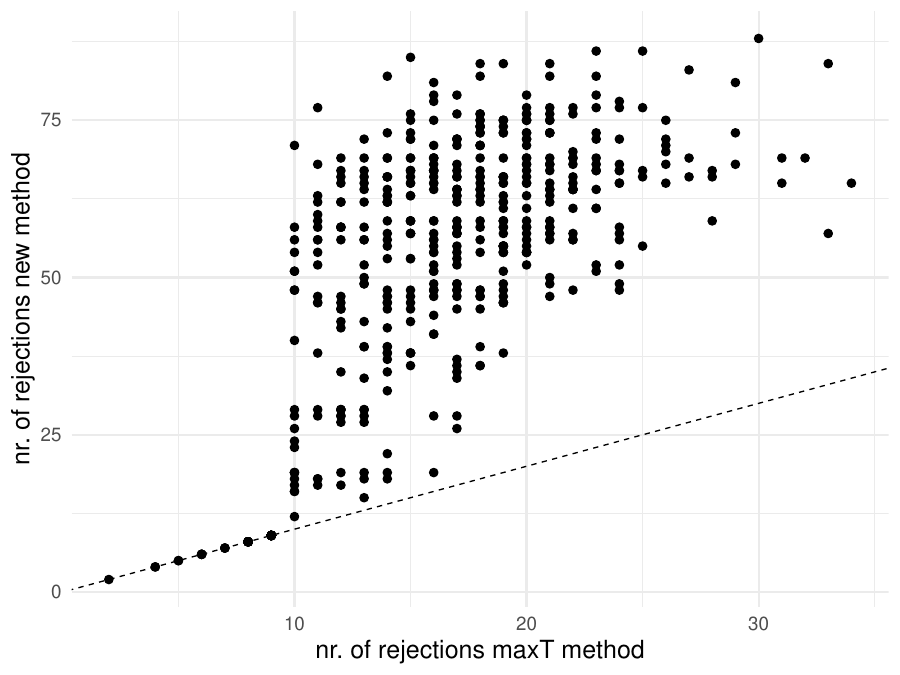}
  \caption{Scatterplot showing the number of rejections of the new method versus the number of rejections of maxT. Each point corresponds to a simulation.
The simulation setting was the same as for Fig. \ref{fig:powervsgamma}. In the new method, $\gamma=0.1$ was taken. MaxT corresponds to $\gamma=0$, by Theorem \ref{thmgamma0}. The plot  illustrates that if the new method rejects fewer than $\gamma^{-1}=10$ hypotheses, then maxT rejects the same, and vice versa.}
 \label{fig:plotnrej}
\end{figure}

\begin{corollary} \label{sameglobaltest}
Consider any $\gamma\in(0,1)$. 
If our single-step method (Section \ref{secsingle}) rejects at least one hypothesis, then with probability 1 so does
single-step maxT and vice versa. The same holds for the sequential versions of the methods.
\end{corollary}
\begin{proof}
The sequential version of each method only rejects something if the single-step version does. Hence, it suffices to prove the result for the single-step versions, which we now consider.
Note that if one of the methods rejects nothing, then by Proposition \ref{samewhenfewrej}, with probability 1, the other method rejects nothing either. Hence, if one method rejects something, the other method does so too with probability 1.
\end{proof}

The statements in Proposition \ref{samewhenfewrej} and Corollary \ref{sameglobaltest} hold with probability 1, i.e., almost surely, rather than surely. Section \ref{reasonprob1} includes a discussion of why this is the case.


\subsection{Remark on Proposition \ref{samewhenfewrej} and Corollary \ref{sameglobaltest}}  \label{reasonprob1}
The statements in Proposition \ref{samewhenfewrej} and Corollary \ref{sameglobaltest} hold with probability 1, i.e., almost surely, rather than surely. The reason is that these results rely on Assumption \ref{asscont}, which implies that with probability 1 there are no ties among the test statistics. This assumption is typically satisfied for continuous data. When this assumption is not satisfied, it can happen that maxT rejects fewer than $\gamma^{-1}$ hypotheses while  the new method rejects more. Indeed, suppose that single-step maxT rejects fewer than $\gamma^{-1}$ hypotheses, i.e., 
\begin{equation} \label{eqmaxrejectsfewer}
|\{1\leq i \leq m: T_i(X)>Q_0\}|<\gamma^{-1}.
\end{equation}
For discrete data, it may happen that multiple statistics are exactly $Q_0$, so that  while \eqref{eqmaxrejectsfewer} holds, at the same we have that 
\begin{equation*} 
|\{1\leq i \leq m: T_i(X)\geq Q_0\}|\geq \gamma^{-1}.
\end{equation*}
It is possible  that there is an $\epsilon > 0$ such that 
 for at least $(1-\alpha)100\%$ of the transformations $g\in \G$,
$$R(Q_0-\epsilon  ,gX)\leq 1\quad \text{and}\quad  R(Q_0  ,gX)=0.$$
In that case, the fact that for all $t< Q_0$ we have $R(t,X)\geq \gamma^{-1}$, implies that 
for at least $(1-\alpha)|\G|$ of the transformations, for all $t\in(Q_0-\epsilon,Q_0)$,
\begin{equation} \label{bordercase1}
\frac{ R(t  ,gX)}{R(t,X)\vee 1}\leq \frac{1}{\gamma^{-1}}= \gamma 
\end{equation}
and for at least $(1-\alpha)|\G|$ of the transformations, for all $t\geq Q_0$, 
\begin{equation} \label{bordercase2}
\frac{ R(t  ,gX)}{R(t,X)\vee 1}= \frac{0 }{R(t,X)\vee 1}=0.
\end{equation}
Inequalities \eqref{bordercase1} and \eqref{bordercase2} imply that for at least $(1-\alpha)|\G|$ of the transformations,
for all  $t> Q_0-\epsilon$, 
$\frac{ R(t  ,gX)}{R(t,X)\vee 1}\leq \gamma$, which means that
$q<Q_0$ and the single-step FDX method rejects at least $\gamma^{-1}$ hypotheses, while maxT rejects strictly fewer than $\gamma^{-1}$ hypotheses. However,  this will  be rare in practice with discrete data.

\section{Closed testing and monotonicity} \label{appct}

Here we discuss the new methods in the context of \citt{goeman2021only}, which provides general theory on multiple testing procedures that make confidence statements on false discovery proportions.
Corollary \ref{sameglobaltest} states that the new method and maxT essentially imply the same global test. This is interesting, because in the sense of Theorem 1 of \citt{goeman2021only}, this would imply that the closed testing procedure corresponding to the new method is the same as the closed testing procedure corresponding to maxT --- which is in fact the sequential maxT procedure itself. Indeed, Theorem 1 of \citt{goeman2021only} states that if a procedure satisfies the so-called \emph{monotonicity} assumption  \citp[][p.1225]{goeman2021only}, then the global test of the procedure implies a corresponding closed testing procedure with certain good properties. However, while most multiple testing methods have the monotonicity property, this is not the case for all FDX methods. In particular, the FDX methods in  \citt{korn2004controlling} and \citt{romano2007control}, as well as the new method, are not generally monotone. For the new method, we illustrate this with a small example below. 
Thus, Theorem 1 of \citt{goeman2021only} does apply to the new method.
In particular, it can easily be seen that the new method provides FDP statements that are not implied by the sequential maxT method. Here, it is important to note that maxT is a \emph{consonant} \citp[][p.1228]{goeman2021only} multiple testing method, which means that it cannot be expanded to make FDP statements except those already implied by its FWER rejections. 

We now provide the counterexample which illustrates that the new method is not monotone. 

\begin{example}[New method is not generally monotone]
For the definition of monotonicity, we refer to \citt{goeman2021only}.
To prove that the new method is not generally monotone, it suffices to give a counterexample where adding a single hypothesis to the problem decreases some of the FDP bounds for the original hypotheses considered.

As a minimal example, consider 6 hypotheses and suppose  that for every hypothesis $\Hy_i$ we have a single observation $X_i$, so the data are $X=(X_1,...X_6)$.. We use sign-flipping \citp{hemerik2020robust}, so $\G=\{id, g\}$ has cardinality 2. Here, $id$ is the identity map and $g$ multiplies the entire data vector by $-1$. Each observation  serves as its own statistic, i.e. we let $T_i(X)=X_i$, so that we we always have $X_i=T_i(X)=T_i(idX)=-T_i(gX).$
Take $\alpha=\gamma=0.5$.
Suppose $0<X_1<X_2<X_3<X_4<T_5(gX)=-X_5<X_6$.
Then $q=\min\{s_{id},s_g\}=s_g=T_5(gX)$, so $\R=\{6\}$, i.e., we only reject one hypothesis. Hence, for example, among $\{1,2,3,4,6\}$, we know that at least one hypothesis is false.

Now we add a 7th hypothesis $\Hy_7$ with $X_7>X_6$. 
Then  $q=s_g\leq 0$, so $\R=\{1,2,3,4,6,7\}$ i.e., we reject six hypotheses --- all except 5. With $(1-\alpha)100\%$ confidence, we know that at least $(1-\gamma)\cdot 6=3$ of these hypotheses are false. Hence, the method tells us that among $\{1,2,3,4,6\}$,  at least two are false. We see that adding the extra hypothesis to the problem, allowed us to say more about the original six hypotheses.
\end{example}

\section{Data analysis: aggregated music review scores} \label{secdataan}
As a simple illustration of the new method, we applied it to a dataset containing aggregated music review scores \citep{musicdata}. 
The dataset contains aggregated album review scores from the  websites \emph{Album of the Year} and \emph{Metacritic}.
We only used the data collected from Album of the Year. 
There are critic scores and user scores; for this illustration, we analysed the user scores. The scores were integers between 0 and 100, so that  they are essentially continuous.
The most recent albums in the data are from the year  2020.
We only selected albums (including EPs) that were released in 1980 or later. We further reduced the data by only selecting artists with at least 10 albums since 1980. This resulted in a dataset containing information on 260 artists, including many  prominent artists from the most popular genres, such as pop (Mariah Carey, Robbie Williams), rock (AC/DC, Pearl jam) and hip-hop (Eminem, Snoop Dogg). For every artist, we only used the data on the first 10 albums since 1980.

Roughly speaking, we are  interested in whether scores of consecutive albums by an artist are simply i.i.d. --- or at least exchangeable --- or whether that is not the case, perhaps due to some trend. We will assume throughout that observations from different artists are independent of each other.
As a first simple analysis, for each of the 260 artists, we computed the Pearson correlation of the album scores with their corresponding release years. We call these correlations $T_1,...,T_{260}$.
The mean of these 260 correlations was $-0.26.$ Heuristically, this suggests that there is an overall negative trend. 
To test whether $T_1,...,T_{260}$ all have mean 0,
we applied a two-sided one-sample t-test to these 260 values, which gave a p-value below $10^{-15}$. 

\subsection{Familywise error rate control}

Whether $\mathbb{E}(T_i)$ has a negative expectation, might depend on the artist. For example, $\mathbb{E}(T_i)$ might be influenced by the genre or home country of the artist.
Hence, we may ask for each artist separately whether their test statistic has mean 0. 
More precisely, we consider hypotheses $\Hy_1,...,\Hy_{260}$, where $\Hy_i$ is the null hypotheses that $\mathbb{E}(T_i)=0$, $1\leq i \leq m$. 
One possible multiple testing approach that controls the familywise error rate, is to  perform individual permutation tests and then apply Bonferroni-Holm to the resulting p-values.
Instead we used the maxT method, based on the absolute correlations $|T_1|,...,|T_m|$. 
The reason is that we expect maxT to have better power than Bonferroni-Holm, as discussed below. We also provide results for Bonferroni-Holm for comparison.

For maxT to have any power, it suffices to use $\lceil \alpha^{-1}\rceil$ permutations. However, to enhance stability, we used $10^4$ permutations.
We randomly permuted the ten order labels $10^4$ times --- including the original order. There are $|G|=10!=3628800$ possible permutations; the reason for using  random permutations was of course the reduction in computation time. The maxT method is valid if Assumption \ref{assjointd} holds. This property immediately follows here, from the assumption of independence between artists in combination with the  definition of our null hypotheses. 
MaxT likely has somewhat better power than Bonferroni-Holm. Indeed, even if the artists are independent, there may still be trends present --- indeed, this is what we wish to test. For example, the first albums of artists may score high on average. When such patterns are present, maxT tends to have better power than Bonferroni-Holm.

The  sequential maxT method rejected 3 hypotheses for $\alpha=0.05$ and 8 hypotheses for $\alpha=0.1$. Bonferroni-Holm based on permutaton tests  rejected only 1 hypothesis for $\alpha=0.05$ and 3 for $\alpha=0.1$.
As an illustration, we consider $\alpha=0.1$. The  sequential maxT method  rejected  the hypotheses corresponding to the artists Pet Shop Boys,    Pearl Jam,   Korn,  Pixies, Van Morrison,   Sting, The Mountain Goats and The Flaming Lips. Six of these artists had a strongly negative test statistic, indicating a clear decline in review scores. Two of these artists, namely the last two mentioned, had a strongly positive test statistic.

\subsection{Multi-resolution FDX control}
Finally, we applied the simultaneous  FDX method from Section \ref{secsingle}. Note that the (theoretically) required Assumption \ref{assjointdextra} is satisfied here, if the artists are independent of each other.
The method rejected 8 hypotheses for $\gamma=0.1$, 22 hypotheses for $\gamma=0.2$, 48  for $\gamma=0.3$ and 95 for $\gamma=0.5$. For example,  in case we had taken $\gamma=0.5$, we would know that with $90\%$ confidence, at least $50\%$ of those 95 hypotheses are false. 
Note that the new method allows ``zooming in''. For example, suppose we had chosen $\gamma=0.2$; then we do not only know with $90\%$ confidence that among the top 22 hypotheses, at most $\gamma100\%=20\%$ are true. Indeed, simultaneously, we know that among the top 9 hypotheses, at most $20\%$, so 1 out of 10, are true. Moreover, we simultaneously know that among the top 4 hypotheses, at most $20\%$ are true, which means that all 4 are false. With $90\%$ confidence, all those statements are simultaneously true, assuming we  chose $\gamma=0.2$ beforehand.

Among the 95 hypotheses that we reject for $\gamma=0.5$,
81 had a strongly negative test statistic, the other 14 a strongly positive statistic.
All in all, the results suggest that not all artists have i.i.d. album scores. Heuristically speaking, 
negative trends seem more common than positive trends, although both seem to exist. There are probably multiple reasons for the overall negative trend.
One reason might be a selection bias: an artist who has a  longstanding  career, probaby had early albums that were appreciated by the public, encouraging the artist and their music label to continue  releasing music. Another reason might be that reviewers tend to rate older albums higher.

\section{Proofs of results} \label{appproofs}

\subsection{Proof of Theorem \ref{thmmain}}

\begin{proof}
We start with an overview of the proof.
\\
\\
\emph{Overview of the proof.}\\
The main strategy of the proof is to define two additional data-dependent thresholds $q^*=q^*(X)$ and $q'=q'(X)$. Computing these thresholds requires knowledge of $\N$, so they cannot be used in practice.
We show that under an event $\E$ with $\pr(\E)\geq 1-\alpha$, we have  
\begin{equation} \label{propertyqstar}
\forall t\geq q^*:  FDP(t)\leq \gamma.
\end{equation}
Further, we show that
\begin{equation} \label{eqorderingqs}
q^*  \,{\buildrel \E \over \leq}\, q' \leq q, 
\end{equation}
where ``$\,{\buildrel \E \over \leq}\,$'' means that ``$\leq$'' holds under $\E$.
Thus, $\E$ does not only imply \eqref{propertyqstar} but also that $q^*\leq q$.
It follows that under $\E$,
$$\forall t\geq q:  FDP(t)\leq \gamma.$$
Since $\pr(\E)\geq 1-\alpha$, this implies the theorem.

In step 1 we will define $q'$ and note that $q'\leq q$. In step 2 we define $q^*$ and $\E$ and show that $\pr(\E)\geq 1-\alpha$. In step 3, we show that \eqref{eqorderingqs} holds, which implies the theorem. Step 4 is essentially a lemma used in step 3. 
\\
\\
\emph{Step 1.}
For $t\in \reals$ define 
$R_{\F}(t,X)=   |\F \cap     \mathcal{R}(t,X)|.$  
For every $g\in \G$, let 
$$s_g'(X) =\sup\Big\{t\in\reals:  \frac{V(t,gX)}{R(t,X)\vee 1}>\gamma\Big\}=
$$
$$   \sup\Big\{t\in\reals:  \frac{V(t,gX)}{(V(t,X)+ R_{\F}(t,X))\vee 1}>\gamma\Big\}=   $$
$$   \sup\Big\{t\in\reals:   V(t,gX)-\gamma V(t,X)     >\gamma  R_{\F}(t,X)\Big\},$$
where the supremum of an empty set is $-\infty$.
Let $q'=q'(X)$ be the $(1-\alpha)$-quantile of the values $s_g'$, $g\in \G$.
Note that $q'\leq q$, since for every $g\in \G$, $s_g'\leq s_g$. In the remainder of the proof, we will show that the threshold $q'$ is valid in the sense that Theorem \ref{thmmain} holds if we replace $q$ by $q'$ there. This immediately implies that the threshold $q$ is also valid, which will finish the proof.
\\
\\
\emph{Step 2.}
For every $g\in \G$, let 
$$s^*_g(X) =  \sup\Big\{t\in\reals:  \frac{V(t,gX)}{(V(t,gX)+ R_{\F}(t,X))\vee 1}>\gamma\Big\}= $$
$$  \sup\Big\{t\in\reals:  V(t,gX)   -\gamma  V(t,gX)         >\gamma  R_{\F}(t,X)   \Big\}.$$
Note the subtle difference compared to the definition of $s_g'(X)$.
Let $q^*=q^*(X)$ be the  $(1-\alpha)$-quantile of the values $s^*_g$, $g\in \G$.

Note that due to right-continuity of $V(\cdot , gX)$, for at least $(1-\alpha)100\%$ of the transformations $g\in \G$, we have that
$$\forall t\geq  q^*(X):  V(t,gX)   -\gamma  V(t,gX)         \leq \gamma  R_{\F}(t,X).$$

Define the event 
$$   \E =\Big\{ \forall t\geq q^*:  V(t,X)   -\gamma  V(t,X)         \leq \gamma  R_{\F}(t,X)       \Big\}=\Big\{ \forall t\geq q^*:  FDP(t,X)\leq \gamma      \Big\}.$$
Note here that $R_{\F}(t,X)$ is function of  $X_{\F}$ and  the values $V(\cdot,gX)$, $g\in\G$, are functions of  $X_{\N}$.

Conditional on $(T_i(X):i\in \F)$,  $(s^*_g:\text{ }g\in \G)$ is a function of the test statistics corresponding to $\N$ and is hence
 invariant under all transformations $g\in \G$, by Assumption \ref{assjointd}.
Hence, conditional on $(T_i(X):i\in \F)$, by the group invariance principle \citp{hemerik2018exact,hemerik2018false}, 
$$\pr(s^*_{id}\leq q^*)\geq 1-\alpha.$$
Note that 
$$s^*_{id}\leq q^* \Longleftrightarrow \Big\{ \forall t\geq q^*:  FDP(t,X)\leq \gamma      \Big\}\Longleftrightarrow \E.$$
This means that $$\pr(\E)\geq 1-\alpha.$$ This is true conditionally, hence also marginally.

Thus, if we use $q^*$ as the threshold, then we have FDX control.
However, we do not know $q^*$ in practice. We will show though that if $\E$ holds, then $q'\geq q^*$ and hence $FDP(t,X)\leq \gamma$ for all $t\geq q'$. This is the aim of the remainder of the proof.
\\
\\
\emph{Step 3.}
Suppose $\E$ holds. Consider any $g\in \G$ for which $s'_g< q^*$. 
In Step 4 below, we will show that under $\E$ it then holds that $s_g^*<q^*$ with probability 1. Taking that for granted now, suppose $q'<q^*$.
Then
$$(1-\alpha)d \leq |\{g\in\G: s_g'\leq q'\}|\leq |\{g\in\G: s_g'< q^*\}| \leq |\{g\in\G: s_g^*< q^*\}| ,$$
which is a contradiction with the fact that $|\{g\in\G: s_g^*< q^*\}|<(1-\alpha)d$. It follows  that  $q'\geq q^*$  under $\E$ as we wanted to show.

Thus, what remains to show is that under $\E$,
if $g\in \G$ is such that  $s_g'< q^*$, then $s_g^* < q^*$ with probability 1. For the rest of the proof, suppose $\E$ holds.
\\
\\
\emph{Step 4: Showing that under $\E$, $s_g' < q^*$ implies  $s_g^* < q^*$ with probability 1.} \\
We will show that $s^*_g\geq q^*$ implies   $s_g' \geq q^*$ with probability 1, which gives the result.

First consider the case that the inequality is strict, i.e., $s^*_g > q^*$.
Then 
there is a $t\geq q^* $ for which 
\begin{equation} \label{eq1}
V(t,gX)   -\gamma  V(t,gX)         >\gamma  R_{\F}(t,X). 
\end{equation}
Consider any such $t$.
Since we assumed $\E$, we  have  that $$V(t,X)   -\gamma  V(t,X) \leq  \gamma  R_{\F}(t,X)   <  V(t,gX)   -\gamma  V(t,gX) $$ and hence $V(t,X) < V(t,gX)$, so that 
\begin{equation} \label{eq2}
    V(t,gX)   -\gamma  V(t,X)    \geq V(t,gX)   -\gamma  V(t,gX).
  \end{equation}
By inequalities \eqref{eq1} and  \eqref{eq2},
$$    V(t,gX)   -\gamma  V(t,X)    >\gamma  R_{\F}(t,X)  $$ 
so that $s_g'>t\geq q^*$.

Finally, consider the case that $s^*_g = q^*$. We must show that $s_g' \geq q^*$. 
Note that if $g=id$, then obviously $s^*_g=s_g'$, so that $s_g'\geq q^*$. Now suppose $g\neq id$. 
That means that at $s^*_g=q^*$, $V(\cdot,gX)$ makes a jump downwards.  
Thus, for all sufficiently small $\eps>0$,
\begin{equation} \label{eq3}
V(q^*-\eps,gX)   -\gamma  V(q^*-\eps,gX)         >\gamma  R_{\F}(q^*-\eps,X).
\end{equation} 
Due to Assumption \ref{asscont},
with probability 1, in a neigbourhood of $q^*$, $V(\cdot,X)$  
has no discontinuity. Hence, since we supposed $\E$, with probability 1, for all sufficiently small $\eps>0$ we have 
\begin{equation} \label{eq4}
V(q^*-\eps,X)   -\gamma  V(q^*-\eps,X)         \leq \gamma  R_{\F}(q^*-\eps,X).
\end{equation} 
Inequalities \eqref{eq3} and \eqref{eq4}  imply that 
 $$V(q^*-\eps,X) < V(q^*-\eps,gX),$$ so that $$-\gamma  V(q^*-\eps,X)\geq    -\gamma  V(q^*-\eps,gX)  $$ and hence, 
 \begin{equation} \label{eq5}
    V(q^*-\eps,gX)   -\gamma  V(q^*-\eps,X)    \geq V(q^*-\eps,gX)   -\gamma  V(q^*-\eps,gX)  
\end{equation} 
for every sufficiently small $\eps>0$.
Now inequalities \eqref{eq3} and \eqref{eq5}  imply that  with probability 1,
$$    V(q^*-\eps,gX)   -\gamma  V(q^*-\eps,X)    >\gamma  R_{\F}(q^*-\eps,X)   $$
for every sufficiently small $\eps>0$.
Hence, with probability 1, $s_g' \geq q^*$, as we wanted, which finishes the proof.

\end{proof}

\subsection{Proof of Theorem \ref{thmseq}}

\begin{proof}
As in the proof of Theorem \ref{thmmain}, consider the event
$$   \E =\Big\{ \forall t\geq q^*:  FDP(t,X)\leq \gamma      \Big\}.$$
In Theorem \ref{thmmain}, we showed that $\pr(\E)\geq 1-\alpha$ and  that under $\E$, $FDP(t)\leq \gamma$ for all $t\geq q'$. Here, we will show that under $\E$, $ q'\leq q_{\fin}$, so that $FDP(t)\leq \gamma$ for all $t\geq q_{\fin}$.

Suppose $\E$ holds. 
Using induction, we will prove that for all $i\in \mathbb{N}$ we have $q'\leq q_i$.
We know that this holds for $i=0$. Now suppose we have proven this for all indices up to $i \in \mathbb{N}$. We must prove that it also holds for index $i+1$. To this end, note that we know that $FDP(q_{i})\leq \gamma$. This means that in the set $\R(q_i)$, there are at least $\lceil(1-\gamma)R(q_i)\rceil=B_{i+1}$ false hypotheses.
 
 This means that there exists a set  $\I\subseteq\{1,...,m\}$ with $\I^c\subseteq \R(q)$ and $|\I^c| = B_i$ such that  $ \I\supseteq\N$.  
 For such a  set $\I$, for every $g\in \G$ we have 
 $$|\I\cap \R(t,gX)| \geq |\N\cap \R(t,gX)|= V(t,gX) $$
  and consequently  
 $s^{\I}_{g,i+1} \geq s_g'$.
 But this means that $q'\leq q^{\I}_{i+1}$, which implies that $q'\leq q_{i+1}$ by definition of $q_{i+1}$.
 It follows by induction that this holds for all $i\in \mathbb{N}$.
  
  The sequence $B_1\leq B_2\leq ...$ is a discrete, bounded sequence, so that it converges after finite number of steps. Hence, the sequence  $q_0,  q_1,  q_2,...$ also converges after a finite number of steps. We conclude that under $\E$, $ q'\leq q_{\fin}$, so that under $\E$, $FDP(t)\leq \gamma$ for all $t\geq q_{\fin}$.
\end{proof}

\subsection{Proof of Proposition \ref{propasympt}}

\begin{proof}
In the proof of Theorem \ref{thmmain}, we only used Assumption \ref{assjointdextra} in Step 2. This assumption guaranteed that the quantities $V(t,gX)$, $g\in \G$, were independent of $R_{\F}(t,X)$. Pick $0<\epsilon<1$. Then, by our asymptotic assumption, we can choose $N$ such that for all $n\geq N$, with probability at least $1-\epsilon$, 
$$  \max\big \{T_i(gX_n): i\in \N,g\in \G'\big \}< \min\big \{T_i(X_n): i\in \F\big \} .$$
Hence, for all $n\geq N$, for all $g\in\G_n'$,  with probability at least $1-\epsilon$,  $s^*_{g}(X) < c$ and $s'_{g}(X) < c$ and hence,   with probability at least $1-\epsilon$, the following equalities hold:
$$s^*_{g}(X_n) = \tilde{s}^*_{g}(X_n) := \sup\Big\{t\in\reals:  \frac{V(t,gX_n)}{(V(t,gX_n)+ |\F|)\vee 1}>\gamma\Big\},$$
$$s'_{g}(X_n) = \tilde{s}'_{g}(X_n) := \sup\Big\{t\in\reals:  \frac{V(t,gX_n)}{(V(t,X_n)+ |\F|)\vee 1}>\gamma\Big\}.$$
Thus, as $n\rightarrow\infty$, 
\begin{equation} \label{eqto1a}
\pr\big\{\forall g'\in \G': s'_{g}(X_n) =  \tilde{s}'_{g}(X_n)\big \}   \rightarrow 1
\end{equation}
and hence 
\begin{equation} \label{eqto1b}
\pr\big\{ q'(X_n) =  \tilde{q}'(X_n)\big \}  \rightarrow 1.
 \end{equation}
where $\tilde{q}' (X_n)$ is the $(1-\alpha)$-quantile of the test statistics $\tilde{s}'_g(X_n)$, $g\in \G'$.

Note that  $\tilde{s}^*_{g}(X_n)$ does not depend  on the statistics $T_i(X_n)$ with $i\in \F$.
Hence,  with Assumption \ref{assjointd} but without Assumption \ref{assjointdextra}, it follows as in  step 2 the proof of Theorem \ref{thmmain} that regardless of $n$,
$$\pr\big\{\tilde{s}^*_{id}(X_n)\leq \tilde{q}^*(X_n)\big \}\geq 1-\alpha$$
and as in step 3 we find that consequently
\begin{equation} \label{eqqtilde}
\pr\big\{\tilde{s}'_{id}(X_n)\leq \tilde{q}'(X_n)\big \}\geq 1-\alpha,
 \end{equation}
where $\tilde{q}'(X_n)$ is the $(1-\alpha)$-quantile of the values $\tilde{s}'_{g}(X_n)$, $g\in\G_n'$.

Combining \eqref{eqqtilde} with  \eqref{eqto1a} and \eqref{eqto1b} gives
$$\liminf_{n\rightarrow\infty}  \pr\big\{{s}'_{id}(X_n)\leq {q}'(X_n)\big \}\geq 1-\alpha.$$
Note that ${s}'_{id}(X_n) = \sup\big\{t\in\reals:  FDP(t,X_n)>\gamma\big\}$, so the above translates to
$$\liminf_{n\rightarrow\infty}  \pr\big\{\forall t\geq q'(X_n): FDP(t,X_n)\leq \gamma    \big \}\geq 1-\alpha.$$
 Since $q'(X_n) \leq q(X_n)$ by definition, the result follows for the single-step method.

Regarding the sequential method, recall that  $q'(X_n) \leq q_{\fin}(X_n)$, so this result also immediately follows.
\end{proof}

\subsection{Proof of Theorem \ref{thmgamma0}}

\begin{proof}
We first prove the result for the single-step method.
Since $\gamma=0$,  for every $g\in\G$ we have
\begin{align*}
s_{g}  
&=\sup\Big\{t\in\mathbb{R}: \frac{R(t,gX)}{R(t,X)\vee 1}>0\Big\} \\
&= \sup\Big\{t\in\mathbb{R}: R(t,gX)>0\Big\}\\
&= \max_{1\leq i \leq m}T_i(gX).
\end{align*}
Our threshold $q$ is the $(1-\alpha)$-quantile of the values $s_{g}$, $g\in \G$.
The threshold $Q_0$ is the $(1-\alpha)$-quantile of the values $\max_{1\leq i \leq m}T_i(gX)$, $g\in \G$.
Hence, $q=Q_0$. Thus, when $\gamma=0$, our single-step FDX method rejects the same  hypotheses as the single-step maxT method.

We now prove the result for the sequential method. 
We will show that $q_{j}=Q_j$ for every $j\geq 0$, which implies $q_{\fin}=Q_{\fin}$, which means the methods are the same.
We proceed by induction. Let $j\geq 0$ and suppose we have proven that 
$q_{j}=Q_j$. We will show that $q_{j+1}=Q_{j+1}$.
To do this, note that 
$B_{j+1}=\lceil (1-\gamma)R(q_j)\rceil = R(q_{j})$, so that $\K_{j+1}$ contains a single element, namely  $\A_{j+1}:=\R^c(q_{j})$, the complement of $\R(q_{j})$.
Consequently, 
for every $g\in \G$,
\begin{align*}
s_{g,{j+1}}^{\A_{j+1}}   
&=\sup\Big\{t\in\reals:  \frac{|\A_{j+1} \cap \R(t,gX)|}{R(t,X)\vee 1}>\gamma\Big\} \\
&= \sup\Big\{t\in\reals:  \frac{|\A_{j+1} \cap \R(t,gX)|}{R(t,X)\vee 1}>0\Big\} \\
&= \sup\Big\{t\in\reals:  |\A_{j+1} \cap \R(t,gX)|>0\Big\} \\
&= \max_{i\in \A_{j+1}} T_i(gX).
\end{align*}
Recall that $q_{j+1}$ is the $(1-\alpha)$-quantile of the values $s_{g,{j+1}}^{\A_{1}} $, $g\in \G$. Further, $Q_{j+1}$ is the $(1-\alpha)$-quantile of the values  $\max_{i\in \A_{j+1}} T_i(gX)$, $g\in \G$. Hence, $q_{j+1} = Q_{j+1}$, which finishes the proof by induction.
\end{proof}

\subsection{Proof of Proposition \ref{samewhenfewrej}}

\begin{proof}
We first prove the result for the single-step methods. We then continue with a proof for the sequential methods.
\\
\\
\emph{Proof for the single-step methods.}
By Theorem \ref{thmgamma0}, if $\gamma=0$, then our method is equivalent to maxT. It is straightforward to show that increasing $\gamma$ can only lead to more rejections. More precisely, $q\leq Q_0$.

It is left to show that if $\gamma\in(0,1)$ and single-step maxT rejects strictly fewer than $\gamma^{-1}$ hypotheses, then $Q_0\leq q$ with probability 1. 
Thus, suppose $R(Q_0)<\gamma^{-1}$.

To prove that $Q_0\leq q$, note that it suffices to show that if $g\in \G$ is such that $ Q_0\leq   \max_{1\leq i \leq m}T_i(gX)  $, then $\max_{1\leq i \leq m}T_i(gX)= s_g(X)$. 
Indeed, that will imply that 
$$|\{g\in \G:  Q_0\leq   \max_{1\leq i \leq m}T_i(gX)  \}| \leq   |\{g\in \G:  Q_0\leq   s_g(X)  \}| ,$$
which means that
$$|\{g\in \G:  s_g(X)< Q_0  \}|     \leq |\{g\in \G:   \max_{1\leq i \leq m}T_i(gX)  < Q_0  \}|  < (1-\alpha)d,  $$
so that $q$ cannot be smaller than $Q_0$.
Thus, consider a $g\in \G$ for which  $ Q_0\leq \max_{1\leq i \leq m}T_i(gX).$

First suppose $g=id$. 
Then $$s_g= \sup \{t\in \reals: R(t,X)>\gamma (R(t,X)\vee 1)\} = \sup \{t\in \reals: R(t,X)>0\}=$$
$$ \max_{1\leq i \leq m}T_i(X)=  \max_{1\leq i \leq m}T_i(gX).$$
as we wanted.

Now suppose $g\neq id$.
By Assumption \ref{asscont}, since $R(\cdot,gX)$ has a discontinuity at $\max_{1\leq i \leq m}T_i(gX)$, it follows that  with probability 1, $R(\cdot,X)$ has no discontinuity there.
Note that 
\begin{equation} \label{eql1}
\max_{1\leq i \leq m}T_i(gX) = \sup\Big\{t\in\reals: R(t,gX)>0\Big\}.
\end{equation}
For all $t\geq \max_{1\leq i \leq m}T_i(gX)$, we have $t\geq Q_0$ and hence $\gamma\cdot (R(t,X)\vee 1)<1$ by the assumption that maxT rejects fewer than $\gamma^{-1}$ hypotheses. Since with probability 1, $R(t,X)$ has no discontinuity in a neighbourhood of $\max_{1\leq i \leq m}T_i(gX)$, there is an $\epsilon>0$ such that for all $t\geq \max_{1\leq i \leq m}T_i(gX) -\epsilon$, $R(t,X)\vee 1=0.$
It follows that with probability 1, \eqref{eql1}  equals 
  $$\sup\Big\{t\in\reals:  R(t,gX)>\gamma\cdot(R(t,X)\vee 1)\Big\}=$$
  $$\sup\Big\{t\in\reals: \frac{R(t,gX)}{R(t,X)\vee 1}>\gamma\Big\}= s_{g}(X),$$
  as we wanted. This finishes the proof that $q=Q_0$ with probability 1.
\\
\\
\emph{Proof for the sequential methods.}
For $\gamma=0$ we have $q_{\fin}=Q_{\fin}$, by the proof of Theorem \ref{thmgamma0}. The threshold $q_{\fin}$ is nonincreasing in $\gamma$, so in general $q_{\fin}\leq Q_{\fin}$. It is left to show that if $R(Q_{\fin})<\gamma^{-1}$, then  $Q_{\fin}\leq q_{\fin}$ with probability 1. 

We have
$0\leq \gamma R(q) < 1$ and consequently   
\begin{equation}  \label{gammaRsmall}
\lceil -\gamma R(q)\rceil =0.
\end{equation}
Recall the notation from Theorem \ref{thmseq}. We have

$$B_1= \lceil (1-\gamma) R(q)\rceil  = \lceil  R(q)  + (-\gamma R(q)) \rceil. $$
Since $R(q)$ is an integer, this equals
$$  R(q)  + \lceil -\gamma R(q)\rceil = R(q)+0$$
 by \eqref{gammaRsmall}. Thus, $B_1$ is simply $R(q)$. It follows that $\K_1$ only contains one element, namely $\A_1=\R^c(q)$.
 
For $\gamma=0$ we have $q_1= Q_1$, by the proof of Theorem \ref{thmgamma0}. For larger $\gamma$, we have $q_1\leq Q_1$. 
Now we will show that $Q_1\leq q_1$. 
To do this, we will show that for every $g\in \G$,
if $\max_{i\in\A_1}T_i(gX)\geq Q_1$, then $\max_{i\in\A_1}T_i(gX) =s_{g,1}^{\A_1}$. It then follows as in the above proof for the single-step case, that  $Q_1\leq q_1$.

Thus, pick a $g\in \G$ for which $\max_{i\in\A_1}T_i(gX)\geq Q_1$.
First consider the case $g=id$. 
By assumption, the sequential maxT method rejects fewer than $\gamma^{-1}$ hypotheses, which implies that 
for all $t\geq Q_1$, $R(t,X) < \gamma^{-1}$. Hence,  for all $t\geq \max_{i\in\A_1}T_i(X)$, $R(t,X) < \gamma^{-1}$. 
By Assumption \ref{asscont}, with probability 1, there are no ties among the test statistics.
Hence, with probability 1, there is an $\epsilon>0$ such that for all $t\geq \max_{i\in\A_1}T_i(X)-\epsilon$, $R(t,X)< \gamma^{-1}$, i.e., $0< \gamma\cdot (R(t,X)\vee 1)<1$.
Consequently, with probability 1,  

\begin{align*}
     \max_{i\in\A_1}T_i(gX) 
        &=  \sup\Big\{t\in\reals:  |\A_{1} \cap \R(t,X)|>0\Big\} \\
        &=  \sup\Big\{t\in\reals:  |\A_{1} \cap \R(t,X)|>\gamma\cdot (R(t,X)\vee 1)\Big\} \\
        &= \sup\Big\{t\in\reals:  \frac{|\A_{1} \cap \R(t,X)|}{R(t,X)\vee 1}>\gamma\Big\} = s_{g,1}^{\A_1} 
\end{align*}
as we wanted.

Now suppose  $g\neq id$.
Then 
$$\max_{i\in\A_1}T_i(gX) =$$
  $$\sup\Big\{t\in\reals:  |\A_{1} \cap \R(t,gX)|>0\Big\}.$$
  For all $t\geq \max_{i\in\A_1}T_i(gX)$, we have $t\geq Q_1\geq Q_{\fin}$. Since by assumption $R(Q_{\fin})<\gamma^{-1}$, it follows that  $0<\gamma\cdot (R(t,X)\vee 1)<1.$
  Like above, note that with probability 1, $R(\cdot,X)$ has no discontinuity at $\max_{i\in\A_1}T_i(gX)$, so there is an $\epsilon>0$ such that for all $t\geq \max_{i\in\A_1}T_i(gX) -\epsilon$, $\gamma\cdot (R(t,X)\vee 1)=0.$  
   Hence with probability 1 the above equals
  $$\sup\Big\{t\in\reals:  |\A_{1} \cap \R(t,gX)|>\gamma(R(t,X)\vee 1)\Big\}=$$
  $$\sup\Big\{t\in\reals:  \frac{|\A_{1} \cap \R(t,gX)|}{R(t,X)\vee 1}>\gamma\Big\}= s_{g,1}^{\A_1}$$
  as we wanted to show. We conclude that $Q_1\leq  q_1$, so $Q_1= q_1$ with probability 1. Showing that $Q_2\leq  q_2$ with probability 1 can be done analogously, so that we find  $q_2= Q_2$ with probability 1. Continuing like this, we find that $q_{\fin}=Q_{\fin}$ with probability 1 as desired.
\end{proof}

\end{document}